%% file: main.tex
\documentclass[12pt]{article}

\newif\iflong
\longtrue

\newif\ifCoRRversion
\CoRRversiontrue
\input{settings}

\input{macros}

\title{Complexity Estimates for Fourier-Motzkin Elimination}
\author{Rui-Juan Jing, Marc
  Moreno Maza and Delaram Talaashrafi}

\begin{document}
\maketitle

\input{abstract}

\input{Introduction}

\input{Background}

\input{Revision}

\input{ProjectedPolyhedron}

\input{complexity}

\input{Experiments}

\input{Relatedwork}

\input{Application}

\input{main.bbl}
\end{document}

%% file: settings.tex
\usepackage{lhelp}
\usepackage{amsfonts,latexsym}
\usepackage{amsmath}
\usepackage{mathtools}
\usepackage{amssymb}
\usepackage{caption}
\usepackage{algorithm}
\usepackage{algorithmic}
\usepackage{verbatim}
\usepackage{tikz}
\usetikzlibrary{matrix,positioning,decorations.pathreplacing,calc}
\usepackage{listings}
\usepackage{hyperref}
\usepackage{cleveref}
\usepackage{color}
\usepackage{cite}
\usepackage{geometry}
\usepackage{comment}
\usepackage{multicol}
\usepackage{paralist}
\usepackage{import}

\usepackage{bibentry}

\usepackage{listings}

\usepackage{courier}
\definecolor{mygreen}{rgb}{0,0.7,0.1}
\definecolor{mygray}{rgb}{0.85,0.85,0.85}
\definecolor{mymauve}{rgb}{0.58,0,0.82}
\definecolor{myblue}{rgb}{0.,0.,1}
\definecolor{myred}{rgb}{0.85,0.2,0.1}

%% file: macros.tex
\def\Q{\mbox{${\mathbb Q}$}}
\def\R{\mbox{${\mathbb R}$}}
\def\Z{\mbox{${\mathbb Z}$}}

\def\A {\ensuremath{\mathbf{A}}}

\def\a {\ensuremath{\mathbf{a}}}
\def\b {\ensuremath{\mathbf{b}}}
\def\c {\ensuremath{\mathbf{c}}}
\def\d {\ensuremath{\mathbf{d}}}
\def\g {\ensuremath{\mathbf{g}}}

\def\p {\ensuremath{\mathbf{p}}}

\def\r {\ensuremath{\mathbf{r}}}
\def\s {\ensuremath{\mathbf{s}}}
\def\t {\ensuremath{\mathbf{t}}}
\def\x {\ensuremath{\mathbf{x}}}
\def\u {\ensuremath{\mathbf{u}}}
\def\v {\ensuremath{\mathbf{v}}}
\def\w {\ensuremath{\mathbf{w}}}
\def\y {\ensuremath{\mathbf{y}}}

\def\e {\ensuremath{\mathbf{e}}}
\def\0 {\ensuremath{\mathbf{0}}}
\def\PP {\mathcal{P}}
\def\FM {Fourier-Motzkin elimination}
\def\history{history set}
\def\bftheta{{\bf \Theta}}
\def\bfalpha{{\bf \alpha}}
\def\bflambda{{\bf \lambda}}

\def\sc{{subsumption cone}}
\def\mpr{{minimal projected representation}}

\newcommand{\rank}[1]{\ensuremath{{\rm rank}(#1)}}

\def\cone{{\sf Cone}}

\def\Char{{\sf CharCone}}

\def\linspace{{\sf LinearSpace}}

\def\hom{{\sf HomCone}}
\def\xl{x_{\rm last}}

\def\combine{{\sf Combine}}

\def\testcone{{redundancy test cone}}
\def\initialtestcone{{initial redundancy test cone}}
\def\height{{\sf height}}
\def\projectcone{{projection cone}}
\def\PR{{\sf ProjRep}}

\def\BC{{\mathcal{B}}}
\def\MC{{\mathcal{M}}}

\ifCoRRversion
\newtheorem{corollary}{Corollary}
\newtheorem{theorem}{Theorem}
\newtheorem{lemma}{Lemma}

\newtheorem{definition}{Definition}
\newtheorem{proposition}{Proposition}
\newtheorem{example}{Example}
\newenvironment{proof}{ {\noindent} {\rm Proof} $\rhd$}{$\lhd$}

\newtheorem{remark}{Remark}
\newtheorem{notation}{Notation}

\fi

\newcommand{\proj}[1]{\mbox{{\sf proj}$(#1)$}}

\newcommand{\hidetext}[1]{\mbox{ \ }}

%% file: abstract.tex
\begin{abstract}
In this paper, we propose a new method for removing all the redundant
inequalities generated by {\FM}. This method is based on  an improved
version of Balas' work and can also be used to remove all the
redundant inequalities in the input system. Moreover, our method only uses
arithmetic operations on matrices and avoids resorting to linear
programming techniques.  Algebraic complexity estimates and
experimental results show that our method outperforms alternative
approaches, in particular those based on linear programming and
simplex algorithm.
\end{abstract}

%% file: Introduction.tex
\section{Introduction}
\label{introduction}

Polyhedral sets play an important role in computational sciences.  
For instance, they are used to model, analyze, transform and schedule
for-loops of computer programs; 
we refer to the articles~\cite{Feautrier91dataflowanalysis,Feautrier:1996:APP:647429.723579,grosser.11.impact,Benabderrahmane:2010:PMM:2175462.2175484,Bondhugula:2008:PAP:1379022.1375595,Bastoul:2004:CGP:1025127.1025992,DBLP:journals/taco/VerdoolaegeJCGTC13}.
Of prime importance are the following operations
on polyhedral sets:
\begin{inparaenum}[(i)]
\item conversion between H-representation and V-representation; and
\item projection, namely {\FM} and block elimination.
 \end{inparaenum}

{\FM} is an algorithmic tool for projecting a polyhedral set onto a
linear subspace.  It was proposed independently by Joseph Fourier and
Theodore Motzkin, in 1827 and in 1936.  The original version of this
algorithm produces large amounts of redundant inequalities and has a
double exponential algebraic complexity.  Removing all these
redundancies is equivalent to giving a minimal representation of the
projection of the polyhedron.  Leonid Khachiyan explained
in~\cite{DBLP:reference/opt/Khachiyan09} how linear programming (LP)
could be used to remove all redundant inequalities, thereby reducing the
cost of {\FM} to singly exponential time; Khachiyan did not, however,
give any running time estimate.
As we shall prove in this paper, rather than using linear programming one may 
use only matrix arithmetic, increasing the theoretical and practical efficiency of {\FM} while
still producing an irredundant representation of the
projected polyhedron.

As mentioned above, the so-called {\em block elimination method} 
is another algorithmic tool to project a polyhedral set.  This method requires
enumeration of the extreme rays of a cone.  Many authors have been working
on this topic, see Nata\'lja V. Chernikova~\cite{chernikova1965algorithm}, 
Herv\'e Le Verge~\cite{le1992note}
and Komei Fududa~\cite{fukuda1996double}.  
Other algorithms for projecting polyhedral sets
remove some (but not all) redundant inequalities with
the help of extreme rays: see the work of 
David A. Kohler~\cite{kohler1967projections}.  
As observed by Jean-Louis Imbert in~\cite{imbert1993fourier}, 
the method he proposed in that paper and that of 
Sergei N. Chernikov in~\cite{chernikov1960contraction} are equivalent.
These methods are very effective in practice, but none of them can remove
all redundant inequalities generated by Fourier-Motzkin Elimination.

Egon Balas proposed in~\cite{balas1998projection} a method to overcome this
latter limitation.  
We found flaws, however, in both his construction and
its proof. A detailed account is included in Section~\ref{sec:relatedwork}.

In this paper, we show how to remove all the redundant inequalities
generated by Fourier-Motzkin Elimination based on an improved version
of Balas' work.  To be more specific, a so-called {\em {\testcone}} is
generated by solving a projection problem for a cone which only one
more inequality and one more variable than the inequality defining the
input polyhedron.  This latter projection is carried out by means of
block elimination.  This
{\initialtestcone} is used to remove all the redundant inequalities in
the input polyhedron.  Moreover, our method has a better algebraic
complexity estimate than the approaches using linear programming;
see~\cite{jing2017integerpoints,jing2017computing} for estimates of
those approaches.

For an input pointed polyhedron $Q \subseteq {\Q}^n$, 
given by a system of $m$ linear inequalities of height $h$, we show 
(see Theorem~\ref{thm:comp}) that eliminating the variables from that
system, one after another (thus performing {\FM}) can be done within
$O(m^{\frac{5n}{2}} n^{\theta + 1 + \epsilon} h^{1+\epsilon})$,
for any $\epsilon > 0$, where $\theta$ is the exponent of linear algebra.
Our algorithm is stated in 
Section~\ref{sec:MRPP}
and follows a revisited version of Balas' algorithm
presented in Section~\ref{sec:revision}.
Since the maximum number of facets of 
any standard projection of $Q$ is  $O(m^{\lfloor n/2 \rfloor})$,
our running time for {\FM} is satisfactory;
the other factors in our estimate come
from the cost of linear algebra operations
for testing redundancy.
\iflong
We have implemented the algorithms
proposed in Section~\ref{sec:MRPP}
using the BPAS library~\cite{DBLP:conf/icms/ChenCMMXX14}
publicly available at \url{www.bpaslib.org}.
\fi
We have compared our code against
other implementations of {\FM}
including the CDD library~\cite{fukudacdd}.
Our experimental results, reported in Section \ref{sec:exp},
show that our proposed method
can solve more test-cases (actually all)
that we used while the counterpart software have
failed to solve some of them.

Section \ref{sec:background} provides background materials about
polyhedral sets and polyhedral cones together with 
the original version of {\FM}.
Section \ref{sec:revision} contains 
a revisited version of Balas'
method and detailed proofs of its correctness.
Based on this,
Section \ref{sec:MRPP} presents a new algorithm
producing a {\em minimal projected representation} for a given
full-dimensional pointed
polyhedron.
Complexity results are established 
in Section \ref{sec:comp}.
In Section \ref{sec:exp} we report on our experimentation
and in Section \ref{sec:relatedwork} we discuss related work.
\ifCoRRversion
Finally, 
\Cref{sec:application} shows an application of {\FM}: solving parametric
linear programming (PLP) problems, which is a core routine
in the analysis, transformation and scheduling of for-loops
of computer programs. 

\fi

%% file: Background.tex
\section{Background}
\label{sec:background}

In this section, we review the basics of polyhedral geometry.
Section~\ref{sec:PolyhedralCones} is dedicated to the notions of
polyhedral sets and polyhedral cones.
Sections~\ref{sec:DDAlgorithm} and 
\ref{FM} review the double description method and 
Fourier-Motzkin elimination, which are two of the most important 
algorithms for operating on polyhedral sets.
We conclude this section with the cost model that we shall use for
complexity analysis, see Section~\ref{costmodel}.  We omit most
proofs.  For more details please refer to
\cite{Schrijver:1986:TLI:17634,terzer2009large,fukuda1996double}.  In
a sake of simplicity in the complexity analysis of the presented
algorithms, we constraint our coefficient field to the rational number
field $\Q$. However, all of the results in this paper generalize to
polyhedral sets with coefficients in the field $\R$ of real numbers.

\input{polyhedralsetsandcones}

\input{DDmethod}

\input{FourierMotzkinElimination}

\input{CostModel}

%% file: polyhedralsetsandcones.tex
\ifCoRRversion
\subsection{Polyhedral cones and polyhedral sets}
\label{sec:PolyhedralCones}

 \begin{notation} 
We use bold letters, e.g. $\v$, to denote vectors and 
we use capital letters, e.g. $A$, to denote matrices. 
Also, we assume that vectors are column vectors. For row 
vectors, we use the transposition notation, that is,  
$A^t$ for the transposition of a  matrix $A$.
For a matrix $A$ and an integer $k$,
$A_k$ is the row of index $k$ in $A$. Also, if $K$ is a set of integers,
$A_K$ denotes the sub-matrix of $A$ with row indices in $K$. 
 \end{notation}
 
We begin this section with the fundamental theorem of linear inequalities.

\begin{theorem}[\cite{Schrijver:1986:TLI:17634}]
Let $\a_1,\cdots,\a_m$ be a set of linearly independent vectors in
$\Q^n$. Also, let $\b$ be a vector in $\Q^n$. Then, exactly one of the
following holds:
\begin{enumerate}
\item[$(i)$] the vector $\b$ is a non-negative linear combination of  
$\a_1,\ldots,\a_m$. In other words, there exist 
positive numbers $ y_1,\ldots,y_m$ such that
we have $\b = \sum_{i=1}^m y_i \a_i$, or,
\item[$(ii)$] there exists a vector $\d \in \Q^n$, such that 
both $\d^t \b < 0$ 
and $\d^t \a_i \ge 0$ hold for all $1 \leq i \leq m$.
\end{enumerate}
\end{theorem}

 \begin{definition}[Convex cone] 
A subset of points $C \subseteq \Q^n$ is called a {\em cone} if for
each $\x \in C$ and each real number $\lambda \geq 0$
we have $\lambda \x  \in C$.
A cone $C \subseteq \Q^n$ is called {\em convex} if for 
all $\x, \y \in C$, we have $\x + \y \in C$.
If $C \subseteq \Q^n$ is a convex cone, then
its elements are called the {\em rays} of $C$.
For two rays $\r$ and $\r'$ of $C$, we write
$\r' \simeq \r$ whenever there exists $\lambda  \geq 0$
such that we have $\r' = \lambda  \r$.
\end{definition} 
 \begin{definition}[Hyperplane]
A subset $H \subseteq \Q^n$ is called a {\em hyperplane} 
if $H = \{\x \in \Q^n \ | \ \a^t\x=0 \}$ for some
non-zero vector $\a \in \Q^n$.
\end{definition} 
 \begin{definition}[Half-space]
A {\em half-space} is a set of the 
form $\{x \in \Q^n\ | \ \a^t x \leq 0\}$
for a some vector $\a \in \Q^n$.
\end{definition} 
\begin{definition}[Polyhedral cone] 
A cone $C \subseteq \Q^n$ is a \textit{polyhedral cone} if it is the intersection of finitely many
half-spaces, that is, $C = \{x\in \Q^n \ | \ Ax \leq \0 \}$ for
some matrix $A \in \Q^{m \times n}$.
 \end{definition} 
 \begin{definition}[Finitely generated cone] 
Let $\{\x_1,\ldots,\x_m\}$ be a set of vectors in $\Q^n$. 
The {\em cone generated} by $\{\x_1,\ldots,\x_m\}$, 
denoted by $\cone(\x_1,\cdots,\x_m)$, 
is the smallest convex cone containing those vectors. 
In other words, we have 
$\cone(\x_1,\ldots,\x_m) = \{\lambda_1 \x_1 + \cdots + \lambda_m\x_m \ 
 | \ \lambda_1 \geq 0, \ldots, \lambda_m \geq 0 \}$.
A cone obtained in this way is called 
a {\em finitely generated cone}.
\end{definition} 
With the following lemma, which is a consequence of the fundamental
Theorem of linear inequalities, we can say that the two concepts of
polyhedral cones and finitely generated cones are
equivalent, see~\cite{Schrijver:1986:TLI:17634}

\begin{theorem}[Minkowski-Weyl theorem]
\label{le:minkowskithm}
A convex cone is polyhedral if and only if it is finitely generated.
\end{theorem}
\begin{definition}[Convex polyhedron] 
A set of vectors $P \subseteq \Q^n$ is called a {\em convex polyhedron} 
if $P = \{\x \ | \ A\x \leq \b \}$,
for a matrix $A \in \Q^{m \times n}$ and a vector $\b \in \Q^m$.
Moreover, the polyhedron $P$ is called a {\em polytope}
if $P$ is bounded.
 \end{definition}
From now on, we always use the notation $P = \{\x \ | \
A\x \leq \b  \}$ to represent a polyhedron in $\Q^n$.
We call the system of linear inequalities
$\{ A \x \le \b \}$  a  {\em representation} of $P$.
\begin{definition}[Minkowski sum]
For two subsets $P$ and $Q$ of $\Q^n$, their {\em Minkowski sum},
denoted by  $P + Q$, is the subset of $\Q^n$ defined as 
$\{p+q \ | \ (p,q) \in P \times Q\}$.
\end{definition}

The following lemma, which is another consequence of the fundamental theorem of linear inequalities,
helps us to determine the relation between polytopes and polyhedra. The proof can be found in 
\cite{Schrijver:1986:TLI:17634}

\begin{lemma}[Decomposition theorem for convex polyhedra]
A subset $P$ of $\Q^n$ is a convex polyhedron if
and only if it can be written as the Minkowski sum of a finitely
generated cone and a polytope.
\end{lemma}

Another consequence of the fundamental theorem of inequalities, is the
famous Farkas lemma. This lemma has different variants.  Here we only
mention a variant from~\cite{Schrijver:1986:TLI:17634}. 

\begin{lemma}[Farkas' lemma]
\label{le:farkaslemma}
  Let $A \in \Q^{m \times n}$ be a matrix and $\b \in \Q^m$ be a
   vector.  Then, there exists a vector $\t \in \Q^n, \ \t \ge \0 $
   satisfying $ A \t = \b$ if and if $\y^t \b \ge 0$ holds for each
   vector $\y \in \Q^m$ such that we have $\y^t A \ge 0$.
\end{lemma}
A consequence of Farkas' lemma is the 
following criterion for testing  
whether an inequality $\c^t \x \le c_0$ is {\em redundant}
w.r.t. a polyhedron representation $A \x \le \b$, 
that is, whether $\c^t \x \le c_0$
is implied by $A \x \le \b$.

\begin{lemma}[Redundancy test criterion]
\label{le:redundancyinpolyhedron}
Let $\c \in \Q^n$, $c_0 \in \Q$, $A \in \Q^{m \times n}$ and $\b \in \Q^m$.
Then, the inequality $\c^t \x \le c_0$ is redundant
w.r.t. the system of 
inequalities $A \x \le \b$  if and only if there exists a
vector $\t \ge \0 $ and a number $\lambda \ge 0$ satisfying $\c^t =
\t^t A $ and $c_0 = \t^t \b + \lambda $.
\end{lemma}

\begin{definition}[Implicit equation] 
An inequality $\ \a ^t \x \leq b$
(with $\a \in \Q^n$ and $b \in \Q$) 
 is an implicit equation 
of the  inequality system $A\x \leq \b$
if $ \a ^t \x = b$ holds for all $\x \in P$.
\end{definition}
  \begin{definition}[Minimal representation]
A representation of a polyhedron is {\em minimal} if 
no inequality of that representation is
implied by the other inequalities of that representation.
 \end{definition}
 \begin{definition}[Characteristic (recession) cone of a polyhedron]
The {\em characteristic cone} of $P$ is the polyhedral cone
denoted by $\Char(P)$ and defined by
$\Char(P) := \{\y \in \Q^n \ |\ \x +\y \in P, \ \forall \x \in P\}
= \{ \y \ | \ A\y \leq \0 \}.$
\end{definition}
\begin{definition}[Linearity space and pointed polyhedron] 
The {\em linearity space} of the polyhedron $P$ 
is the linear space denoted by $\linspace(P)$ and defined as 
$\Char(P) \cap -\Char(P) = \{\y \ | \ A\y = \0 \}$,
where $-\Char(P)$ is the set of the $-\y$ for $\y \in \Char(P)$.
The polyhedron $P$ is {\em pointed} if its linearity space is
$\{  \0 \}$.
\end{definition}
\begin{lemma}[Pointed polyhedron criterion]
The polyhedron $P$ is pointed if and only if 
the matrix $A$ is full column rank.
\end{lemma}

\begin{definition}[Dimension of a polyhedron] 
The {\em dimension} of the polyhedron $P$, denoted by $\dim(P)$, is $n-r$, 
where 
$n$ is dimension\footnote{Of course, this notion
of dimension coincides with the topological one,
that is, the maximum dimension of a ball contained in $P$.}
 of the ambient space 
(that is, ${\Q}^n$) and 
$r$ is the maximum number of implicit equations defined by
linearly independent vectors. 
We say that $P$ is {\em full-dimensional}
whenever $\dim(P) = n$ holds. In another words,
$P$ is full-dimensional if and only if it does not have any implicit equations.
 \end{definition} 

  \begin{definition}[Face of a polyhedron]
A subset $F$ of the polyhedron $P$ is called a 
{\em face} of $P$ if $F$
equals 
$\{\x \in P \ | \ A_{\text{sub}} \x = \b_{\text{sub}} \}$
for a sub-matrix $A_{\text{sub}}$ of $A$ and a sub-vector $\b_{\text{sub}}$
of $\b$.
\end{definition} 
\begin{remark}
It is obvious that every face of a polyhedron is also a polyhedron.
Moreover, the intersection of two faces $F_1$
and $F_2$ of $P$
is another face $F$, which is either $F_1$, or $F_2$, 
or a face with a dimension less than $\min (\dim(F_1),\dim(F_2))$.
Note that $P$ and the empty set are faces of $P$.
\end{remark}
\begin{definition}[Facet of a polyhedron]
A face of $P$, distinct from $P$ and of maximal dimension is called a {\em facet}
of $P$.
\end{definition}
\begin{remark} 
It follows from the previous remark that 
$P$ has at least one facet and that the dimension
of any facet of $P$ is equal to $\dim(P) -1$.
When $P$ is full-dimensional, there is a one-to-one correspondence
between the inequalities in a minimal representation of $P$ and the
facets of $P$. From this latter observation, we deduce that the
minimal representation of a full dimensional polyhedron is unique up
to multiplying each of the defining inequalities by a positive
constant.
\end{remark}
\begin{definition}[Minimal face]
A non-empty face that does not contain any other face of a polyhedron is called
a {\em minimal face} of that polyhedron. 
Specifically, if the polyhedron $P$ is pointed, each
minimal face of $P$ is just a point and is called an \textit{extreme point} 
or \textit{vertex} of $P$.
 \end{definition} 

\begin{definition}[Extreme rays] 
Let $C$ be a cone such that $\dim(\linspace(C)) = t$. 
Then, a face of $C$ of dimension $t+1$
is called a {minimal proper face} of $C$. 
In the special case of a pointed cone, 
that is, whenever $t = 0$ holds,
the dimension of a minimal proper face is $1$ 
and such a face is called an {\em extreme ray }. 
We call an {\em extreme ray} of the polyhedron $P$
any extreme ray of its characteristic 
cone $\Char(P)$.
We say that two extreme rays $\r$ and $\r'$ of the polyhedron $P$
are {\em equivalent}, and denote it by $\r \simeq \r'$,
if one is a positive multiple of the other.
When we consider the set of all extreme rays
      of the polyhedron $P$ (or the polyhedral cone $C$)
     we will only consider one
     ray from each equivalence class.
\end{definition}
\begin{lemma}[Generating a cone from its extreme rays] 
\label{le:genwithextr}
A pointed cone $C$ can be generated by its extreme rays, that is,
we have
$C \ = \ \{\x \in \Q^n \ | \   ({ \exists \c \geq \0  }) \    \x = R \c\}$,
where the columns of $R$ are the extreme rays of $C$.
\end{lemma}

\begin{remark} 
\label{rem:extremerays}
From the previous definitions and lemmas, we derive
the following observations:
\begin{enumerate}
\item the number of extreme rays of each cone is finite,
\item the set of all extreme rays is unique up to multiplication by a scalar, and,
\item all members of a cone are positive linear combination of extreme rays.
\end{enumerate}
\end{remark}

We denote by ${\sf ExtremeRays}(C)$ the set of extreme rays 
of the cone $C$. Recall that all cones considered here
are polyhedral.

The following, see~\cite{mcmullen1970maximum,terzer2009large},
is helpful in the analysis of algorithms manipulating extreme rays of
cones and polyhedra.

\begin{lemma}[Maximum number of extreme rays]
\label{le:maxextr}

Let $E(C)$ be the number of extreme rays of a polyhedral cone $C \in {\Q}^n$ 
with $m$ facets. Then, we have:
\begin{equation}
\label{eq:EC}
E(C) \leq \binom{m-\lfloor{\frac{n+1}{2}}\rfloor}{m-1} + \binom{m- \lfloor{\frac{n+2}{2}}\rfloor}{m-n} 
\leq m^{\lfloor{\frac{n}{2}}\rfloor}.
\end{equation}
\end{lemma}

From Remark~\ref{rem:extremerays}, it appears that extreme rays are
important characteristics of polyhedral cones.  Therefore, two
algorithms have been developed in~\cite{fukuda1996double} to check
whether a member of a cone is an extreme ray or not.  For explaining
these algorithms, we need the following definition.

 \begin{definition}[Zero set of a cone] 
For a cone $C = \{\x \in \Q^n \ | \ A \x \leq \0 \}$ and $\t \in C$,
we define the {\em zero set} $\zeta_A(\t)$ as 
the set of row indices $i$ such that $A_i \t = 0$, where $A_i$ is the
$i$-th row of $A$.
For simplicity, we use $\zeta(\t)$ instead of $\zeta_A(\t)$ when there
is no ambiguity.
 \end{definition}
Consider a cone $C = \{\x \in \Q^n \ | \ A^\prime\x = \0 , \
A^{\prime \prime}\x \leq \0 \}$ where $A^\prime$ and
$A^{\prime \prime}$ are two matrices such that the system
$A^{\prime \prime}\x \leq \0 \ $ has no implicit equations.  
The proofs of the following lemmas are straightforward and can be
found in \cite{fukuda1996double} and \cite{terzer2009large}.
\begin{lemma}[Algebraic test for extreme rays]
\label{le:algetest} 
Let $\r \in C$. Then, the ray
$\r$ is an extreme ray of $C$ if and only if
we have
$\mathrm{rank}\left(\left[\begin{aligned} & A^{\prime} \\ & A_{\zeta(r)}^{\prime \prime} \end{aligned} \right]\right) = n - 1$.
 \end{lemma} 
\begin{lemma}[Combinatorial test for extreme rays]
\label{le:combtest}
Let  $\r \in C$. Then, the ray 
$\r$ is an extreme ray of $C$ if and only if for any ray $\r'$ of
$C$ such that $\zeta(\r) \subseteq \zeta(\r')$ holds
we have $\r' \simeq \r$.
\end{lemma}

\begin{definition}[Polar cone] 
\label{def:polarcone} 
For the given polyhedral cone $C \subseteq \Q^n$, the {\em polar cone} induced 
by $C$ is denoted $C^*$ and given by:
 \begin{center} 
$C^* = \{\y \in \Q^n \ | \ \y^t \x \leq \0 , \forall \x \in C \}.$
 \end{center} 
\end{definition}
The following lemma shows an important property of the polar cone
of a polyhedral cone. The proof can be found in \cite{Schrijver:1986:TLI:17634}.
 \begin{lemma}[Polarity property]
\label{le:polarconeproperty} 
For a given cone $C \in \Q^n$, there is a one-to-one correspondence
between the faces of $C$ of dimension $k$ and the faces of $C^*$ of dimension
$n-k$. In particular, there is a one-to-one correspondence between
the facets of $C$ and the extreme rays of $C^*$.
 \end{lemma}
Each polyhedron $P$ can be embedded in a higher-dimensional cone,
called the homogenized cone associated with $P$.

\begin{definition}[Homogenized cone of a polyhedron]
The {\em homogenized cone} 
of the polyhedron $P = \{\x \in \Q ^ n \ | \ A\x \leq \b \}$
is denoted  by $\hom(P)$ and defined by:
\begin{center} 
$\hom(P) = \{(\x, \xl) \in \Q^{n + 1} \ | \ C [\x^t,\xl]^t \leq 0\}$,
 \end{center}
where
\[
C = 
\begin{bmatrix}
A & -\b \\
\0 ^t & -1
\end{bmatrix}
\] 
is an $(m+1) \times (n+1)$-matrix, if $A$ is an $(m \times n)$-matrix.
 \end{definition}

\begin{lemma}[H-representation correspondence]
  \label{le:homfacet}
An inequality $A_i \x \leq b_i $ is redundant in $P$
if and only if the corresponding inequality $A_i \x - b_i \xl \leq 0$
is redundant in $\hom(P)$.
\end{lemma}

\begin{theorem}[Extreme rays of the homogenized cone]
\label{le:homextr}
Every extreme ray of the homogenized cone $\hom(P)$
associated with the polyhedron $P$ 
is either of the form
$(\x,0)$ where $\x$ is an extreme ray of $P$, or $(\x,1)$ where $\x$ is an
extreme point of $P$.  
\end{theorem}

\fi

%% file: DDmethod.tex
\ifCoRRversion
\subsection{Polyhedral computations}
\label{sec:PolyhedralComputations}

In this section, we review two of the most important algorithms
for polyhedral computations: the double description algorithm (DD for short) and
the Fourier-Motzkin elimination algorithm (FME for short).

A polyhedral cone $C$ can be represented either as an intersection of
finitely many half-spaces (thus using the so-called {\em H-representation}
of $C$) or as by its extreme rays (thus using the so-called
{\em V-representation} of $C$); the DD algorithm produces one representation
from the other.  We shall explain the version of the DD algorithm
which takes as input the H-representation of $C$ and returns as output
the V-representation of $C$.

The FME algorithm performs a standard projection
of a polyhedral set to lower dimension subspace. 
In algebraic terms, this algorithm takes
as input a polyhedron $P$ given by a 
system of linear inequalities (thus an H-representation of $P$) 
in $n$ variables $x_1 < x_2 < \cdots < x_n$ and computes
the H-representation of the projection of $P$ on 
$x_1 < \cdots < x_k$ for some $1 \leq k < n$.

\subsubsection{The double description method}
\label{sec:DDAlgorithm}
We know from \Cref{le:minkowskithm} that any polyhedral cone 
$C = \{ \x \in \Q^n \ | A \x \le \0 \}$ 
can be generated by finitely many vectors, say
$\{\x_1, \ldots, \x_q \} \in \Q^n$. 
Moreover, from \Cref{le:genwithextr} we know that 
if $C$ is pointed, then it can be generated by its extreme rays, 
that is, $C = \cone(R)$ where $R = [\x_1, \ldots, \x_q]$.  
Therefore, we have two possible representations for 
the pointed polyhedral cone $C$: 
\begin{description}
\item[{H-representation}:]
as the intersection of finitely many half-spaces, or equivalently,
with a system of linear inequalities $A \x \le \0 $;
\item[{V-representation}:] 
as a linear combination of finitely many vectors, 
namely $\cone(R)$, where $R$ is a
matrix, the columns of which are the extreme rays of $C$.
\end{description}
We say that the pair $(A,R)$ is a {\em Double Description Pair} or
simply a {\em DD pair} of $C$.  We call $A$ a {\em representation
matrix} of $C$ and $R$ a {\em generating matrix} of $C$.  We call $R$
(resp. $A$) a {\em minimal generating (resp. representing) matrix}
when no proper sub-matrix of $R$ (resp. $A$) is generating
(resp. representing) $C$.

It is important to notice that, for some queries in polyhedral
computations, the output can be calculated in polynomial time using
one representation (either a representation matrix or a generating
matrix) while it would require exponential time using the other
representation.

For example, we can compute in polynomial time the intersection of two
cones when they are in H-representation but the same problem would be
harder to solve when the same cones are in V-representation.
Therefore, it is important to have a procedure to convert between
these two representations, which is the focus of the
articles~\cite{chernikova1965algorithm} and~\cite{terzer2009large}.

We will explain this procedure, which is known as the 
\textit{double description method} as well as \textit{Chernikova's algorithm}. 
This algorithm takes a cone in H-representation as input and returns a
V-representation of the same cone as output.  In other words, this procedure
finds the extreme rays of a polyhedral cone, given by its representation
matrix.
It has been proven that this procedure runs in single exponential
time.  To the best of our knowledge, the most practically efficient
variant of this procedure has been proposed by Fukuda
in~\cite{fukuda1996double} and is implemented in the \textsc{CDD} library. 
We shall explain his approach here and analyze its algebraic 
complexity. 
Before presenting Fukuda's algorithm, we need a few more definitions
and results.  In this section, we assume that the input cone $C$ is
pointed.

\input{DDmethod_algorithm}

\fi

%% file: DDmethod_algorithm.tex
The \textit{double description method} works in an incremental manner. 
Denoting by $H_1, \ldots, H_m$ the half-spaces corresponding
to the inequalities of the H-representation of $C$, we have $C =
H_1 \, \cap \, \cdots \, \cap \, H_m$.  Let $1 < i \leq m$ and assume
that we have computed the extreme rays of the cone $C^{i-1} :=
H_1 \, \cap \, \cdots \, \cap \, H_{i-1}$.  Then the $i$-th iteration
of the DD method deduces the extreme rays of $C^{i}$ from those of
$C^{i-1}$ and $H_{i}$.

Assume that the half-spaces $H_1, \ldots, H_m$ are numbered
such that $H_i$ is given by $A_i\x \leq 0$, where $A_i$
is the $i$-th row of the representing matrix $A$.
We consider the following partition of $\Q^n$:

$H^+_i = \{\x \in \Q^n \ | \ A_i \x > 0 \}$,
$H^0_i = \{\x \in \Q^n \ | \ A_i \x = 0 \}$
and 
$H^-_i = \{ \x \in \Q^n \ | \ A_i \x < 0 \}$.

Assume that we have found the DD-pair $(A^{i-1},R^{i-1})$ 
of $C^{i-1}$.
Let $J$ be the set of the column indices of $R^{i-1}$.
We use the above partition $\{H^+_i, H^0_i, H^-_i \}$
to partition $J$ as follows:

$J^+_i = \{j \in J \ | \ \r_j \in H^+ \}$,
$J^0_i = \{j \in J \ | \ \r_j \in H^0\}$
and
$J^-_i = \{j \in J \ | \ \r_j \in H^- \}$,
\newline
where $\{ \r_j \ | \ j \in J \}$ is the set of the columns
of $R^{i-1}$, hence the set of the extreme rays of $C^{i-1}$.

For future reference, let us denote by $\sf{partition}(J,A_i)$ 
the function which returns $J^+, J^0, J^-$ as defined above.
The proof can be found in \cite{fukuda1996double}.

\begin{lemma}[Double description method]
\label{lem:DDmth1}
Let $J' := J^+ \cup J^0 \cup (J^+ \times J^-)$.
Let $R^i$ be the $(n \times |J'|)$-matrix consisting of
\begin{itemize}
\item the columns of $R^{i-1}$ with index
      in $J^+ \cup J^0$, followed by
\item the vectors ${{\r}'}_{(j,j')}$
      for $(j,j') \in (J^+ \times J^-)$, 
where
\begin{center}
${{\r}'}_{(j,j')} = (A_i \r_j)\r_{j'} - (A_i \r_{j'})\r_{j},$
\end{center}
\end{itemize}
Then, the pair $(A^i,R^i)$ is a DD pair of $C^{i}$.
\end{lemma}

The most efficient way to start the incremental process is to choose
the largest sub-matrix of $A$ with linearly independent rows; we call
this matrix $A^0$. Indeed, denoting by $C^0$
the cone with $A^0$ as representation matrix,
the matrix $A^0$ is invertible and its inverse
gives the extreme rays of $C^0$, that is:
\begin{center}
${\sf ExtremeRays}(C^0) = {(A^{0})}^{-1}$.
\end{center}
Therefore, the first DD-pair that the above
incremental step should take as input is $(A^0,{(A^{0})}^{-1})$.

The next key point towards a practically efficient DD method is to
observe that most of the vectors ${{\r}'}_{(j,j')} $ in \Cref{lem:DDmth1}
are redundant. Indeed, Lemma~\ref{lem:DDmth1} leads to
a construction of a generating matrix of $C$
(in fact, this would be Algorithm~\ref{alg:DDmethod}
where Lines 13 and 16 are suppressed) producing a 
double exponential number  of rays (w.r.t. the ambient dimension $n$) 
whereas
\ifCoRRversion
Lemma~\ref{le:maxextr}
\fi
guarantees that 
the number of extreme rays of a polyhedral cone is
singly exponential in its ambient dimension.
To deal with this issue of redundancy, we need the notion
of \textit{adjacent} extreme rays.

\input{adjacentextremerays}

Based on \Cref{prop:adjtest}, 
we have \Cref{alg:adjtest} for testing whether 
two extreme rays are adjacent or not.

\begin{algorithm}
\caption{AdjacencyTest}
\label{alg:adjtest}
\begin{algorithmic}[1]
\STATE \textbf{Input:} $(A,\r,\r')$, where $A \in \Q ^{m \times n}$ is the representation 
matrix  of cone $C$, $\r$ and $\r'$ are two extreme rays of $C$ 
\STATE \textbf{Output:} true if $\r$ and $\r'$ are adjacent, false otherwise
\STATE $\s := A \r, \ \s' := A \r'$
\STATE let $\zeta(\r)$ and $\zeta(\r')$ be set of indices of zeros in $\s$ and $\s'$ respectively
\STATE $\zeta := \zeta(\r) \cap \zeta(\r')$
\IF{$\rank{A_{\zeta}} = n-2$}
\STATE {\bf return} true 
\ELSE
\STATE {\bf return} false 
\ENDIF
\end{algorithmic}
\end{algorithm}

The following lemma explains how to obtain $(A^i,R^i)$ from
$(A^{i-1},R^{i-1})$, where $A^{i-1}$ (resp. $A^{i}$) is the sub-matrix
of $A$ consisting of its first $i-1$ (resp. $i$) rows.
The double description method is a direct application
of this lemma, see~\cite{fukuda1996double} for details.
\begin{lemma}
\label{lem:ddstrong}
 As above, let $(A^{i-1},R^{i-1})$ be a DD-pair and denote 
by $J$ be the set of indices of the columns of $R^{i-1}$. 
Assume that  $\rank{A^{i-1}} = n$ holds. 
Let $J' := J^- \cup J^0 \cup {\rm Adj}$,
where ${\rm Adj}$ is the set of the 
pairs $(j,j') \in J^+\times J^- $ such that
$\r_j,$ and $\r_{j'}$ are adjacent as extreme rays
of $C^{i-1}$, the cone with $A^{i-1}$
as representing matrix.
Let $R^i$ be the $(n \times |J'|)$-matrix consisting of
\begin{itemize}
\item the columns of $R^{i-1}$ with index
      in $J^- \cup J^0$, followed by
\item  the vectors ${{\r}'}_{(j,j')}$
      for $(j,j') \in (J^+ \times J^-)$, 
where
\begin{center}
${{\r}'}_{(j,j')} = (A_i \r_j)\r_{j'} - (A_i \r_j')\r_{j},$
\end{center}
\end{itemize}
Then, the pair $(A^i,R^i)$ is a DD pair of $C^{i}$.
Furthermore, if $R^{i-1}$ is a minimal generating matrix for the
representation matrix $A^{i-1}$, then $R^{i}$ is 
also
a minimal generating matrix for the representation matrix $A^{i}$.
\end{lemma}

Using \Cref{prop:adjtest} and \Cref{lem:ddstrong} we can obtain \Cref{alg:DDmethod}
\footnote{In this algorithm, $A^i$ shows the representation matrix in step $i$}
 for computing the extreme rays of a cone.

\begin{algorithm}
\caption{DDmethod}
\label{alg:DDmethod}
\begin{algorithmic}[1]
\STATE \textbf{Input: }{ a matrix $A \in \Q^{m \times n}$, a 
representation matrix of a pointed cone $C$ }
\STATE \textbf{Output: }{$R$, the minimal generating matrix of $C$}
\STATE let $K$ be the set of indices of $A$'s independent rows
\STATE $A^0 := A_K$ 
\STATE $R^0 := {(A^{0})}^{-1}$
\STATE let $J$ be set of column indices of $R^0$
\WHILE{$K \neq \{1, \cdots, m \}$}
\STATE select a $A$-row index $i \not \in K$
\STATE $J^+,\ J^0, \ J^- := {\sf partition}(J,A_i)$
\STATE add vectors with indices in $J^+$ and $J^0$ as columns to $R^i$
\FOR{$p \in J^+$}
\FOR{$n \in J^-$}
\IF{${\rm AdjacencyTest}(A^{i-1},\r_p,\r_n) = {\rm true}$}
\STATE $\r_{\rm new} := (A_i \r_p)\r_{n} - (A_i \r_n)\r_p$
\STATE add $\r_{\rm new}$ as columns to $R^i$
\ENDIF
\ENDFOR
\ENDFOR
\STATE let $J$ be set of indices in $R^i$
\ENDWHILE 
\end{algorithmic}
\end{algorithm}

%% file: adjacentextremerays.tex
 \begin{definition}[Adjacent extreme rays] \
Two distinct extreme rays $\r$ and $\r'$ of the polyhedral cone
$C$ are called \textit{adjacent} if they
span a $2$-dimensional face of $C$. 
\footnote{We do not use the minimal face, as it used in the main
reference because it makes confusion.} 
\end{definition}
The following lemma shows how we can test whether two extreme rays are adjacent or not. 
The proof can be found in \cite{fukuda1996double}.
\begin{proposition}[Adjacency test]
\label{prop:adjtest} 
Let $\r$ and $\r'$ be two distinct rays of $C$. 
Then, the following statements are equivalent:
\begin{enumerate}
\item $\r$ and $\r'$ are adjacent extreme rays,
\item $\r$ and $\r'$ are extreme rays and $\rank{A_{\zeta(\r) \cap \zeta(\r')}} = n-2$,
\item if $\r''$ is a ray of $C$ 
      with $\zeta(\r) \cap \zeta(\r') \subseteq \zeta(\r'')$, 
      then $\r''$ is a positive multiple of either $\r$ or $\r'$.
\end{enumerate}
It should be noted that the second statement is related  
to algebraic test for extreme rays while the
third one is related to the combinatorial test.
\end{proposition}

%% file: FourierMotzkinElimination.tex
\ifCoRRversion
\subsubsection{Fourier-Motzkin elimination}
\label{FM}
\begin{definition}[Projection of a polyhedron]
\label{defi:polyProoj}
Let $A \in \Q^{m \times p}$ and $B \in \Q^{m \times q}$
be matrices.
Let $\c \in \Q^m$ be a vector.
Consider the polyhedron
$P \subseteq \Q^{p+q}$ defined by
$P = \{(\u,\x) \in \Q^{p+q} \ | \ A\u + B\x \leq \c \}$.
We denote by $\proj{P; \x}$ the 
\textit{projection of $P$ on $\x$}, that is, the subset of $\Q^q $
defined by
\begin{center} 
$\proj{P;\x} = \{\x \in \Q^q \ | \ \exists \ \u \in \Q^p , \ (\u,\x) \in P \}.$ 
\end{center} 
\end{definition}
Fourier-Motzkin elimination (FME for short) is an algorithm computing
the projection $\proj{P;\x}$ of the polyhedron of $P$ by successively
eliminating the $\u$-variables from the inequality system $A\u +
B\x \leq \c $. This process shows that $\proj{P;\x} $ is also a
polyhedron.
\begin{definition}[Inequality combination]
\label{defi:InequalityCombination}
Let $\ell_1, \ell_2$ be two inequalities: $a_1x_1 + \cdots +
a_nx_n \leq d_1$ and $b_1x_1 + \cdots + b_n x_n \leq d_2$.  Let
$1 \leq i \leq n$ such that the coefficients $a_i$ and $b_i$ of $x_i$
in $\ell_1$ and $\ell_2$ are respectively positive and negative.
The {\em combination} of $\ell_1$
and $\ell_2$ w.r.t. $x_i$, denoted by $\combine(\ell_1,\ell_2,x_i)$, is:
\begin{center}
$-b_i(a_1x_1 + \cdots + a_nx_n) + a_i(b_1x_1 + \cdots + b_n x_n) 
\leq -b_i d_1 + a_i d_2. $
\end{center}
\end{definition}

Theorem~\ref{lem:fmeorig} shows how to compute $\proj{P;\x}$
when $\u$ consists of a single variable $x_i$.
When $\u$ consists of several  variables, FME
obtains the projection $\proj{P;\x}$ by repeated
applications of Theorem~\ref{lem:fmeorig}.

\begin{theorem}[Fourier-Motzkin theorem  \cite{kohler1967projections}]
\label{lem:fmeorig}
Let $A \in \Q^{m \times n}$ be a matrix and let
$\b \in \Q^m$ be a vector.
Consider the polyhedron 
$P = \{ \x \in \Q^n \ | \ A \x  \leq \b \}$.
Let $S$ be the set of inequalities defined by
$A\x \leq \b$. Also, let $1 \leq i \leq n$. 
We partition $S$ according to the sign 
of the coefficient of $x_i$:
$S^+ = \{\ell \in S \ | \ \mathrm{coeff}(\ell,x_i) > 0 \}$,
$S^- = \{\ell \in S \ | \ \mathrm{coeff}(\ell,x_i) < 0 \}$ and
$S^0 = \{\ell \in S \ | \ \mathrm{coeff}(\ell,x_i) = 0 \}$.
We construct the following system of linear inequalities:
\begin{center}
$ S^\prime = \{ 
\combine(s_p,s_n,x_i) \ | \ 
         (s_p, s_n) \in S^+ \times S^- \}
\ \cup \ S^0.$
\end{center}
Then, $S^\prime$ is a representation of $\proj{P;\x \setminus \{ x_i \}}.$
\end{theorem}

With the notations of Theorem~\ref{lem:fmeorig}, assume
that each of $S^+$ and $S^-$ counts $\frac{m}{2}$ inequalities.
Then, the set $S^\prime$ counts $(\frac{m}{2})^2$ inequalities.
After eliminating $p$
variables, the projection  would be given by $O( {(\frac{m}{2})^{2^{p}}}  )$
inequalities. Thus, FME is \textit{double exponential} in $p$.

On the other hand, from \cite{monniaux2010quantifier}
and \cite{jing2017computing}, we know that the maximum number of facets of
the projection on ${\Q}^{n-p}$ of a polyhedron in ${\Q}^{n}$
with $m$ facets is  $O(m^{\lfloor n/2 \rfloor})$.
Hence, it can be concluded that most of the generated
inequalities by 
FME are \textit{redundant}. Eliminating
these redundancies is the main subject of the subsequent 
sections.

\fi

%% file: CostModel.tex
\subsection{Cost model}
\label{costmodel}
We use the notion of {\em height} of an algebraic number
as defined by Michel Waldschmidt in 
Chapter 3 of~\cite{waldschmidt2000diophantine}.
In particular, for any rational number $\frac{a}{b}$, thus with $b \neq 0$,  
 we define the
{\em height} of $\frac{a}{b}$, 
denoted as $\height(\frac{a}{b})$, as 
$\log \max (|a|,|b|)$.
For a given matrix $A \in \Q^{m \times n}$,
let $\| A \|$ denote the infinite norm of $A$, that is, the maximum absolute
value of a coefficient in $A$.
We define the height of $A$, denoted by
$\height(A) := \height(\| A \|)$, as the maximal height of a coefficient in $A$.
For the rest of this section, our main reference
is the PhD thesis of Arne Storjohann~\cite{storjohann2013algorithms}.
Let $k$ be a non-negative integer.
We denote by  $\MC(k)$ an upper
bound for the number of bit operations
required for performing any of the 
basic operations 
(addition, multiplication, division with reminder)
on input $a, b \in \Z$ with $|a|, |b| < 2^k$.
Using the multiplication algorithm
of Arnold Sch{\"{o}}nhage and
               Volker Strassen~\cite{DBLP:journals/computing/SchonhageS71}
one can choose $\MC(k) \in O(k \log k \log \log k)$.

We also need complexity estimates for some matrix operations.  For
positive integers $a,b,c$, let us denote by $\MC \MC(a,b,c)$ 
an upper bound for 
the number of arithmetic operations (on the coefficients) required 
for multiplying an $(a \times b)$-matrix by an $(b \times c)$-matrix.
In the case of square matrices of order $n$,
we simply write $\MC \MC(n)$ instead of $\MC \MC(n,n,n)$.
We denote by $\theta$ the exponent of linear algebra,
that is, the smallest real positive number such that
$\MC \MC(n) \in O (n^{\theta} )$.

In the following, we give complexity estimates 
in terms of $\MC(k) \in O(k \log
k \log \log k)$ and $\BC(k) = \MC(k) \log k \in O(k (\log
k)^2 \log \log k)$.
We replace every term of the form $(\log k)^p (\log \log k)^q
(\log\log \log k)^r$, (where $p, q, r$ are positive real numbers) with
$O(k^{\epsilon})$ where ${\epsilon}$ is a (positive) infinitesimal.
Furthermore, in the complexity estimates of algorithms operating on
matrices and vectors over $\mathbb{Z}$, we use a parameter $ \beta $,
which is a bound on the magnitude of the integers occurring during the
algorithm.
Our complexity estimates are measures in terms of 
\textit{machine word operations.}
Let $A \in \Z^{m \times n}$ and $B \in \Z^{n \times p}$. Then, the
product of $A$ by $B$ can be computed within
$O(\MC \MC(m,n,p)(\log \beta) + (mn+np+mp)\BC (\log \beta))$ word
operations, where $\beta = n \, \|A\| \, \|B\|$ and $\|A\|$
(resp. $\|B\|$) denotes the maximum absolute value of a coefficient in
$A$ (resp. $B$).  Neglecting log factors, this estimate
becomes $O(\max(m,n,p)^\theta \max(h_A,h_b))$ where $h_A = \height(A)$
and $h_B = \height(B)$.
For a matrix $A \in \Z^{m \times n}$, a cost estimate of
Gauss-Jordan transform is $O (nmr^{\theta-2}(\log \beta) + nm(\log
r)\BC (\log \beta))$ word operations, where $r$ is the rank of the
input matrix $A$ and $\beta = (\sqrt{r} \|A\|)^r$.
Letting $h$ be the height of $A$,
for a matrix $A \in \Z^{m \times n}$, with height $h$, computing the rank of $A$ is done within 
       $O(mn^{\theta + \epsilon}h^{1 + \epsilon})$ word operations,
       and computing the inverse of $A$ (when this matrix
      is invertible over $\Q$ and $m = n$) is done within 
$O(m^{\theta + 1 + \epsilon}h^{1 + \epsilon})$ word operations.
Let $A \in \Z^{ n \times n}$  be an integer matrix, which is
invertible over $\Q$. Then, the absolute value of any coefficient in
$A^{-1}$ (inverse of $A$) can be bounded above by $(\sqrt{n-1}\|A\|^{(n-1)})$.

%% file: Revision.tex
\section{Revisiting Balas' method}
\label{sec:revision}
As recalled in
Section \ref{sec:background}, FME produces a representation of
the projection of a polyhedron 
by eliminating one variable atfer another.  However, this
procedure generates lots of redundant inequalities limiting 
its use in practice to polyhedral sets with a handful of variables only.
In this section, we propose an efficient algorithm which generates a minimal 
representation of a full-dimensional pointed polyhedron, as well 
as its projections.
Through this section, we use 
$Q$ to denote a full-dimensional
pointed polyhedron in $\Q^n$,
where
\begin{equation}\label{eq:Q}
  Q = \{(\u,\x) \in \Q^p \times \Q^q \ | \ A\u + B\x \leq \c \},
\end{equation}
with $A \in \Q^{m \times p}$, $B \in \Q^{m \times q}$ and $\c \in
\Q^{m}$.
Thus, $Q$ has no implicit equations in its representation and the coefficient
matrix $[A,
  B]$ has full column rank.
Our goal in this section is to compute the minimal representation of the projection
$\proj{Q;\x}$ given by
\ifCoRRversion
\begin{equation}\label{eq:projectedpolyhedron}
  \proj{Q;\x} := \{ \x \ | \ \exists \u, s.t. (\u, \x) \in Q \}.
\end{equation}
\fi
We call the cone
\ifCoRRversion
\begin{equation}\label{eq:projectioncone}
  C := \{\y \in \Q^m \ | \ \y^tA = 0  \ \ {\rm and} \  \ \y \geq \0
  \}
\end{equation}
\fi
the {\em projection cone} of $Q$ w.r.t.$\u$.
When there is no ambiguity, we simply call $C$ as the projection cone of
$Q$.
Using the following so-called {\em projection lemma}, we can
compute a representation for the projection $\proj{Q; \x}$:
  \begin{lemma}[\cite{chernikov1960contraction}]
\label{le:projectwithextremeray}
The projection $\proj{Q;\x}$ of the polyhedron $Q$ can be represented by
$$ S := \{ \y^tB\x \leq \y^t \c, \forall \y \in \sf{ExtremeRays}(C) \}, $$
where $C$ is the projection cone of $Q$ defined
\ifCoRRversion
by Equation~(\ref{eq:projectioncone}).
\fi
  \end{lemma}

  \Cref{le:projectwithextremeray} provides the main idea of the block
  elimination method.
However, the represention produced in this way may have redundant inequalities.
\ifCoRRversion
The following example from \cite{huynh1992practical} shows this point.
\begin{example}
  Let $P$ be the polyhedron represented by
  \begin{equation}
  P:=  \left\{
    \begin{aligned}
      12 x_1 + x_2 - 3 x_3 + x_4& \le 1 \\
      -36 x_1 - 2 x_2 + 18 x_3 - 11 x_4 & \le -2 \\
      -18 x_1 - x_2 + 9 x_3 - 7 x_4 & \le -1 \\
      45 x_1 + 4 x_2 - 18 x_3 + 13 x_4 & \le 4 \\
      x_1 & \ge 0 \\
      x_2 & \ge 0.
    \end{aligned}
    \right.
  \end{equation}
  The projection cone of $P$ w.r.t. $\u := \{ x_1, x_2 \}$ is
  \begin{equation}
    C:=\left \{
    \begin{aligned}
      12 y_1 - 36 y_2 -18 y_3 + 45 y_4 &= 0,\\
      y_1 -2 y_2 - y_3 + 4 y_4 & = 0, \\
      y_1 \ge 0, y_2 \ge 0, y_3 \ge 0, y_4 & \ge 0.
    \end{aligned}
    \right.
  \end{equation}
  The extreme rays of the cone $C$ are:
  $$(0, 0, 5, 2, 0, 3), (3, 0, 2, 0, 0,
  1), (0, 0, 0, 1, 45, 4), (1, 0, 0, 0, 12, 1), (0, 5, 0, 4, 0, 6),
  (3, 1, 0, 0, 0, 1).$$
  These extreme rays generate a representation of $\proj{P; \{x_3,
    x_4\}}$:
  \begin{equation}
    \left\{
    \begin{aligned}
      3 x_3 - 3 x_4 \le 1, & \ \ \ 9 x_3 - 11 x_4 \le 1, &  \ 6 x_3 -x_4 \le 2, \\
      -3 x_3 + x_4 \le 1, & \ -18 x_3 + 13 x_4  \le 4, & \ 9 x_3 - 8 x_4 \le 1. \\
    \end{aligned}
    \right.
  \end{equation}
  One can check that, in the above system of 
linear inequalities, the inequality $3 x_3 - 3 x_4 \le 1$ is redundant.
\end{example}

\fi
In~\cite{balas1998projection}, 
Balas observed that if the matrix $B$ is invertible,
then we can find a cone such that its extreme rays are in
one-to-one correspondence with the facets of the projection of the
polyhedron (the proof of this fact is similar to the proof of our Theorem 
\ref{thm:polarcone}). Using this fact, Balas developed an algorithm
to find all redundant inequalities for all cases, including the cases
where $B$ is singular.

It should be noted that, although we are using his idea, we have found some
flaws in Balas' paper.
In this section, we will explain the corrected form of Balas' algorithm.
To achieve this, we lift the polyhedron $Q$ to a space in
higher dimension by constructing the following objects.

\par \noindent {\bf \small \underline{ Construction of  $B_0$.}}
  Assume that the first $q$ rows of $B$, denoted as $B_1$,
  are independent.
  Denote the last $m-q$ rows of $B $ as $B_2$.
  Add $m-q$ columns, $\e_{q+1}, \ldots, \e_{m}$, to $B$, where $\e_i$
  is the $i$-th vector in the canonical basis of $\Q^m$, 
thus with $1$ in the $i$-th
  position and $0$'s anywhere else.
  The matrix $B_0$ has the following form:
  $$B_0 = \left[ \begin{aligned} B_1 & \quad \0 \\ B_2 &
      \quad I_{m-q} \end{aligned} \right].$$

To maintain consistency in the notation, let $A_0 = A$ and
  $\c_0 = \c$.

\par \noindent {\bf \small \underline{ Construction of  $Q^0$.}}
 We define:
  \begin{equation*}
    \label{eq:Q0}
    Q^0 := \{(\u,{\x^{\prime}}) \in \Q^p \times \Q^m
    \ | \ A_0\u + B_0\x^\prime \leq \c_0 \ , \ x_{q+1}=\cdots=x_{m} =
    0\}.
  \end{equation*}
  Here and after, we use $\x^{\prime}$ to represent the vector $\x \in
  \Q^q$, augmented with $m-q$ variables ($x_{q+1},\ldots,x_m$).
  Since the extra variables $(x_{q+1},\ldots,x_m)$ are assigned 
  to zero, we note that $\proj{Q;\x}$ and
  $\proj{Q^0;{\x^{\prime}}}$ are ``isomorphic'' by means of the bijection
  $\Phi$:
  \begin{center}
    $ \Phi: \begin{aligned}
       \proj{Q;\x} & \rightarrow \proj{Q^0;\x'} \\
      (x_1, \ldots, x_q) & \mapsto (x_1, \ldots, x_q, 0, \ldots, 0)
    \end{aligned}$
  \end{center}
  In the following, we will treat $\proj{Q;\x}$ and $\proj{Q^0;{\x'}}$
  as the same polyhedron when there is no ambiguity.
 
\par \noindent {\bf \small \underline{ Construction of  $W^0$.}}
Define $W^0$ to be the set of 
all $(\v , \w , v_0)\in \Q^q \times \Q^{m-q} \times
\Q$ satisfying
\begin{equation} \label{eq:cone}
  \{ (\v , \w , v_0)\ | \ [\v^t, \w^t]B_0^{-1}A_0 = 0,
     [\v^t, \w^t]B_0^{-1} \geq 0,
     -[\v^t, \w^t]B_0^{-1}\c_0+v_0 \ge 0  \}.
\end{equation}

This construction of $W^0$ is slightly different from the one in
Balas' work \cite{balas1998projection}. Indeed, we changed
$-[\v^t, \w^t]B_0^{-1}\c_0+v_0 = 0 \ {\rm to} \
-[\v^t, \w^t]B_0^{-1}\c_0+v_0 \ge 0$.
Similar to the discussion in Balas' work, the extreme rays
of the cone $\proj{W^0; \{ \v, v_0 \}}$ are used to construct the
minimal representation of the projection $\proj{Q; \x}$.
To prove this relation, we need a preliminary observation.

\begin{lemma}
 \label{le:projcharcone}
The operations ``computing the characteristic cone'' and
``computing projections'' commute.
To be precise, we have:
$\Char(\proj{Q; \x}) = \proj{\Char(Q); \x}$.
\end{lemma}

\iflong
\input{projection_charcone_commute_proof}
\fi

Theorem \ref{thm:polarcone} shows that extreme rays of the cone
$\overline{\proj{W^0; \{ \v, v_0 \} }}$, which is defined as
$$\overline{\proj{W^0; \{ \v, v_0 \} }}:=\{ (\v, -v_0) \ | \ (\v, v_0) \in \proj{W^0; \{\v, v_0\}} \},$$
are in
one-to-one correspondence with the facets of $\hom( \proj{Q; \x})$ and
as a result its extreme rays can be used to find the minimal
representation of $\hom( \proj{Q; \x})$.

\begin{theorem} 
\label{thm:polarcone}
  The polar cone of $\hom( \proj{Q; \x})$ is equal to $\overline{\proj{W^0; \{
      \v, v_0 \} }}$.
  \end{theorem}

\iflong
\input{polarcone_theorem_proof}
\fi

\begin{theorem}
\label{thm:correspondence}
  The minimal representation of $\proj{Q; \x}$ is given exactly by
  \begin{center}
  $ \{ \v^t \x \le v_0 \ | \ (\v, v_0) \in
  \sf{ExtremeRays}(\proj{W^0; (\v, v_0)}) \setminus \{({\bf 0}, 1) \}\}.$
  \end{center}
 \end{theorem}
 
\begin{proof}
  By Theorem \ref{thm:polarcone}, a minimal representation
  of the homogenized cone
  $\hom(\proj{Q;\x})$ is given exactly by
  $\{  \v \x - v_0 \xl \le 0 \ | \ (\v, v_0) \in \sf{ExtremeRays} 
  (\proj{W^0; (\v, v_0)}) \}$.
  By \Cref{le:homfacet}, any minimal representation of $\hom(\proj{Q;\x})$ has at 
  most one more inequality than any minimal representation of $\proj{Q;\x}$.
  This extra inequality would be $\xl \geq 0$ and, in this case, 
  $\proj{W^0; (\v, v_0)}$ would have the extreme ray $(\0 , 1)$, which can
  be detected easily. Therefore, a minimal representation of 
$\proj{Q; \x}$ is given by
 $ \{ \v^t \x \le v_0 \ | \ (\v, v_0) \in
  \sf{ExtremeRays}(\proj{W^0; (\v, v_0)})\setminus \{({\bf 0}, 1)\} \}.$
\end{proof}


For simplicity, we call the cone $\proj{W^0; \{ \v, v_0 \}}$ the {\em \testcone} of $Q$ w.r.t. $\u$ and denote it
by $\PP_{\u}(Q)$. 
When $\u$ is empty, we define $\PP(Q) := \PP_{\u}(Q)$ and we call it the {\em \initialtestcone}.
\iflong
If there is no ambiguity, we use only $\PP_{\u}$ and $\PP$ to denote the {\testcone} and the {\initialtestcone}, respectively.
\fi
It should be noted that $\PP(Q)$ can be used to detect redundant
inequalities in the input system, as it is shown in 
Steps \ref{ln:inputredbegin} to \ref{ln:inputredend} of Algorithm
\ref{alg:mpr}.

%% file: projection_charcone_commute_proof.tex
\begin{proof}
  By the definition of the characteristic cone, we have $\Char(Q) = \{ (\u, \x) \ |
  \ A \u + B \x \le \0 \}$, whose representation has the same
  left-hand side as the one of  $Q$.
  The lemma is valid if we can show that the representation of
  $\proj{\Char(Q); \x}$ has the same left-hand side as
  $\proj{Q; \x}$.
 This is obvious with the {\FM} procedure.
\end{proof}

%% file: polarcone_theorem_proof.tex
\begin{proof}
  By definition, the polar cone $(\hom (\proj{Q; \x})^*$ is equal to
 { $$\{ (\y, y_0) \ | \ [\y^t, y_0]
  [\x^t, \xl]^t \le 0 , \forall \ (\x, \xl) \in \hom( \proj{Q; \x})
   \}.$$}
 This claim follows immediately from: 
 $(\hom (\proj{Q; \x})^*  = \overline{\proj{W^0; \{ \v, v_0 \} }}$.
 We shall prove this latter equality in two steps.

 $(\supseteq)$ 
 For any $(\overline{\v}, -\overline{v}_0) \in \overline{\proj{W^0; \{
     \v, v_0 \} }}$,
we need to show that $[\overline{\v}^t, -\overline{v}_0][\x^t, \xl]^t \le
0$ holds whenever we have $(\x, \xl) \in \hom(\proj{Q; \x})$.
Remember that we assume that $Q$ is pointed.
Observe that $\hom(\proj{Q; \x})$ is also
pointed. Therefore, we only need to verify the desired property 
for the
extreme rays of $\hom(\proj{Q; \x})$, which either 
have the form $(\s, 1)$ or are equal to $(\s, 0)$ (Theorem~\ref{le:homextr}).
Before continuing, we should notice that
since  $(\overline{\v}, \overline{v}_0) 
 \in \proj{W^0; \{ \v, v_0 \}}$, there exists $\overline{\w}$ such
 that
 $\{[\overline{\v}^t, \overline{\w}^t] B_0^{-1}A_0 = 0, 
   -[\overline{\v}^t, \overline{\w}^t] B_0^{-1}\c_0+\overline{v}_0 \ge
   0,
   [\overline{\v}^t, \overline{\w}^t] B_0^{-1} \geq 0\}.$
   Cases 1 and 2 below conclude that $(\overline{\v}, -\overline{v}_0) \in
\hom(\proj{Q; \x})^*$ holds.
 
Case 1:
For the form $(\s,1)$, we have $\s \in \proj{Q; \x}$.
Indeed, $\s$ is an extreme
point of $\proj{Q; \x } $.
Hence, there exists $\overline{\u} \in \Q^{p}$, such that we have 
$A \overline{\u} + B \s \le \c$. By construction of $Q^0$, we have 
$A_0 \overline{\u} + B_0 \s' \le \c_0$, where $\s' = [\s^t , s_{q+1},
  \ldots,s_m]^t$ with $s_{q+1} = \cdots = s_m = 0$.
Therefore, we have: $[\overline{\v}^t , \overline{\w}^t]B_0^{-1} A_0
\overline{\u} + [\overline{\v}^t , \overline{\w}^t]B_0^{-1} B_0 \s'
\le [\overline{\v}^t , \overline{\w}^t]B_0^{-1} \c_0$.
This leads us to $\overline{\v}^t
\s = [\overline{\v}^t , \overline{\w}^t]\s' \le [\overline{\v}^t , \overline{\w}^t]B_0^{-1} \c_0 \le \overline{v}_0$. 
Therefore, we have $[\overline{\v}^t, -\overline{v}_0][\s^t, \xl]^t \le
0$, as desired.

Case 2:
For the form $(\s, 0)$,  we have $\s \in
\Char(\proj{Q; \x}) = \proj{\Char(Q); \x}$.
Thus, there exists
$\overline{\u} \in \Q^{p}$ such that $A \overline{\u} + B \s \le
\0 $.
Similarly to Case 1, we have $[\overline{\v}^t , \overline{\w}^t]B_0^{-1} A_0 \overline{\u} + [\overline{\v}^t , \overline{\w}^t]B_0^{-1} B_0 \s' \le [\overline{\v}^t , \overline{\w}^t]B_0^{-1} \0 $.  Therefore, we have $\overline{\v}^t \s
= [\overline{\v}^t , \overline{\w}^t]\s' \le [\overline{\v}^t ,
  \overline{\w}^t]B_0^{-1} \0 = 0 $, and thus, 
we have $[\overline{\v}^t, -\overline{v}_0][\s^t, \xl]^t \le
0$, as desired.

$(\subseteq)$ 
  For any 
$(\overline{\y}, \overline{y}_0) \in  \hom(\proj{Q; \x})^*$, we have 
$[\overline{\y}^t,\overline{y}_0][\x^t, \xl]^t \leq 0$ whenever we have 
  $(\x,\xl) \in  \hom(\proj{Q; \x})$.
  For any $\overline{\x} \in \proj{Q; \x}$,  we have
$\overline{\y}^t \overline{\x} \leq -\overline{y}_0$ since 
  $(\overline{\x},1) \in
  \hom(\proj{Q; \x})$.
 Therefore, we have $\overline{\y}^t \x \leq -\overline{y}_0$,
for all $\x \in \proj{Q; \x}$, which makes the inequality $\overline{\y}^t \x \leq -\overline{y}_0$ redundant in the system
$\{A\u + B\x \leq \c \}$.
By  Farkas' Lemma (see Lemma~\ref{le:redundancyinpolyhedron}), there exists $\p \ge \0 , \p \in \Q^m$ and
$\lambda \ge 0$ such that $\p^t A = \0 $, $\overline{\y} = \p ^t B$,
$\overline{y}_0
  = \p^t \c + \lambda$.
Remember that $A_0 = A$, $B_0 = [B, B']$, $\c_0 = \c$.
Here $B'$ is the last $m-q$ columns of $B_0$ consisting of $\e_{q+1},
\ldots, \e_m$. Let $\overline{\w} = \p^t B'$. We then have
$$\{ \p^t A_0 = \0 , \ [\overline{\y}^t, \overline{\w}^t] = \p^t B_0, -\overline{y}_0 \ge \p ^t \c_0,
\p \ge \0 \},$$ which is equivalent to
$$\{\p^t = [\overline{\y}^t, \overline{\w}^t] B_0^{-1}, [\overline{\y}^t, \overline{\w}^t] B_0^{-1} A_0 = \0 , -\overline{y}_0 \ge [\overline{\y}^t,
  \overline{\w}^t] B_0^{-1} \c_0, [\overline{\y}^t, \overline{\w}^t] B_0^{-1} \ge \0 \}.$$
Therefore, $(\overline{\y}, \overline{\w}, -\overline{y}_0) \in W^0$,
which leads us to $(\overline{\y}, -\overline{y}_0) \in \proj{W^0; \{ \v, v_0 \}}$. 
From this, we deduce that $(\overline{\y}, \overline{y}_0) \in \overline{\proj{W^0; \{ \v, v_0 \} }}$ holds.
\end{proof}

%% file: ProjectedPolyhedron.tex
\section{Minimal representation of the projected polyhedron}
\label{sec:MRPP}

In this section, we present our algorithm for removing all the
redundant inequalities generated during Fourier-Motzkin elimination.
Our algorithm detects and eliminates redundant inequalities, right
after their generation, using the {\testcone} introduced in
Section~\ref{sec:revision}.
Intuitively, we need to construct the cone $W^0$ and obtain a
representation of the {\testcone} $\proj{W^0; \{ \v, v_0 \}}$, each
time we eliminate a variable during FME.  This method is time
consuming because it requires to compute the projection of $W^0$ onto
$\{ \v, v_0 \}$ space at each step.  However, as we prove in
\Cref{le:initialprojection}, we only need to compute the
     {\initialtestcone}, using \Cref{alg:itc}, and the {\testcone}s,
     used in the subsequent variable eliminations, can be found
     incrementally without any extra cost.

Note that a byproduct of this algorithm is the {\em minimal projected
  representation} of the input system, according to the specified
variable ordering. This representation is useful for finding solutions
of linear inequality systems.  The projected representation was
introduced in~\cite{jing2017integerpoints,jing2017computing} and will
be reviewed in \Cref{le:projection}.

For convenience, we rewrite the input polyhedron $Q$ defined in
\Cref{eq:Q} as: $ Q = \{ \y \in \Q^{n} \ | \ \A \y \le \c\}$, where
$\A = [A, B]\in \Q^{m \times n}$, $n=p+q$ and $\y = [\u^t, \x^t]^t \in
\Q^n$.  We assume the first $n$ rows of $\A$ are linearly independent.

\begin{algorithm}
  \caption{Generate \initialtestcone}
  \label{alg:itc}
  \begin{algorithmic}[1]
    \REQUIRE{$S = \{ \A \y \le \c \}$: a representation of the input
      polyhedron $Q$;}
    \ENSURE{$\PP$: a representation of the \initialtestcone of $Q$}
    \STATE Construct $\A_0$ in the same way we constructed $B_0$, that is, $\A_0
:= [\A, \A']$, where $\A' := [\e_{n+1}, \ldots, \e_{m}]$ with
    $\e_i$ being the $i$-th vector of the canonical basis of $\Q^{m}$;
    \label{ln:a0}
\STATE
Let $W := \{ (\v, \w, v_0 ) \in \Q^{n} \times \Q^{m-n} \times \Q \ | \ - [\v^t, \w^t]\A_0^{-1} \c + v_0 \ge
0,  [\v^t, \w^t]\A_0^{-1} \ge \0 \}$;
\label{ln:W}
\STATE $\PP = \proj{W; \{ \v, v_0 \}}$;
  \label{ln:projwithblockelimination}
 \RETURN $\PP$
  \end{algorithmic}
\end{algorithm}

\begin{remark}\label{re:proj}
There are two important points about \Cref{alg:itc}.
First, we only need a representation of the 
{\initialtestcone} this representation needs not to be minimal. 
Therefore,  calling \Cref{alg:itc} in Algorithm~\ref{alg:mpr}
(which computes a minimal projected representation
of a polyhedron) does not lead to a recursive
call to Algorithm~\ref{alg:mpr}.
Second, to compute the projection $\proj{W; \{ \v, v_0 \}}$, we need
to eliminate $m-n$ variables from $m+1$ inequalities.  The block
elimination method is applied to achieve this.  As it is shown in
\Cref{le:projectwithextremeray}, the block elimination method will
require to compute the extreme rays of the projection cone (denoted by
$C$), which contains $m+1$ inequalities and $m+1$ variables.
However, considering the structural properties of the coeffient matrix
of the representation of $C$, we found that computing the extreme rays
of $C$ is equivalent to computing the extreme rays of another simplier
cone, which still has $m+1$ inequalities but only $n+1$ variables.
For more details, please refer to Step 3 of
\ifCoRRversion
\Cref{le:compcrtc}.
\fi

\end{remark}
\ifCoRRversion

Lemma~\ref{le:initialprojection} shows how to obtain the {\testcone} $\PP_{\u}$
of the polyhedron $Q$ w.r.t. $\u$ from its {\initialtestcone} $\PP$.
This gives a very cheap way to generate all the {\testcone}s of $Q$
once its {\initialtestcone} is generated; this will be used in Algorithm~\ref{alg:mpr}.
To distinguish from the construction of $\PP$, 
we rename the variables $\v, \w, v_0$ as
$\v_{\u}, \w_{\u}, v_{\u}$, when
constructing $W^0$ and computing the test cone $\PP_{\u}$.
That is, we have $\PP_{\u} = \proj{W^0; \{ \v_{\u}, v_{\u} \}}$, where $W^0$
is the set of
all $(\v_{\u} , \w_{\u} , v_{\u}) \in \Q^q \times \Q^{m-q} \times
\Q$ satisfying
$$\{  (\v_{\u}, \w_{\u}, v_{\u} ) \ | \ 
       [\v_{\u}^t, \w_{\u}^t] B_0^{-1}A = \0 , -[\v_{\u}^t,
  \w_{\u}^t] B_0^{-1} \c + v_{\u} \ge 0, [\v_{\u}^t, \w_{\u}^t]
       B_0^{-1} \ge \0 \},$$
       while we have $\PP = \proj{W; \{ \v, v_0\}}$ where $W$ is the set of all $(\v , \w , v_0)\in \Q^n \times
       \Q^{m-n} \times \Q$ satisfying
       $\{ (\v, \w, v_0 )  \ | \ - [\v^t, \w^t]\A_0^{-1} \c + v_0 \ge
       0,  [\v^t, \w^t]\A_0^{-1} \ge \0 \}.$
       \fi
       \begin{lemma}\label{le:initialprojection}
  Representation of the {\testcone} $\PP_{\u}$  can be obtained from $\PP$ by setting coefficients of the corresponding
  $p$ variables of $\v$ to $0$ in the representation of $\PP$.
\end{lemma}

\iflong
\input{initial_projection_proof}
\fi


For the polyhedron $Q$, given a variable order $y_1> \cdots > y_{n}$,
for $1\le i \le n$, 
we denote by $Q^{(y_i)}$ the
  inequalities in the representation of $Q$ whose largest variable is
  $y_i$.

\begin{definition}[Projected representation]
\label{le:projection}
  The {\em projected representation} of $Q$ $w.r.t.$ the
  variable order $y_1> \cdots > y_{n}$, denoted $\PR(Q; y_1>
  \cdots > y_{n})$, is the linear system given by $Q^{(y_1)}$ if $n=1$, and
  is the conjunction of $Q^{(y_1)}$ and $\PR(\proj{Q; {\y_2}}; y_2 >
  \cdots > y_n)$ otherwise.  
We say that $P := \PR(Q; y_1>
  \cdots > y_{n})$  is a 
{\em minimal projected representation} if, for all $1 \le k \le n$,
every inequality of $P$ with $y_k$ 
as largest variable is not redundant among
all the inequalities of $P$ with variables among 
$y_k, \ldots, y_n$.
\end{definition}
We can generate the \textit{{\mpr}} of a polyhedron by \Cref{alg:mpr}.

\begin{algorithm}
	\caption{RedundancyTest}
	\label{alg:checkextremeray}
	\begin{algorithmic}[1]
          \REQUIRE{$(\PP, \ell)$: where \begin{inparaenum}[(i)] \item
              $\PP := \{ (\v, v_0) \in \Q^n \times \Q \ | \ M [\v^t, v_0]^t \le \0 \}$ with $M \in
              \Q^{m \times {(n+1)} }$, \item
              $\ell : \a^t \y \le c$ with $\a \in \Q^{n}$ and
              $c \in \Q$;
          \end{inparaenum}}
          \ENSURE{false if $[\a^t, c]^t$ is an extreme ray of $\PP$, true otherwise}
	  \STATE Let $M$ be the coefficient matrix of $\PP$ 
		\STATE Let $\s := M [\a^t, c]^t$
		\STATE Let $\zeta(\s)$ be the index set of the zero coefficients of $\s$
		\IF{$\rank{M_{\zeta(\s)}} = n$} 
		\STATE return false
		\ELSE 
		\STATE return true
		\ENDIF
	\end{algorithmic}
\end{algorithm}

\begin{algorithm}[!h]
  \caption{Minimal Projected Representation of $Q$}
  \label{alg:mpr}
  \begin{algorithmic}[1]
    \REQUIRE{$S = \{ \A \y \le \c \}$: a representation of the input
      polyhedron $Q$;}
    \ENSURE{A {\mpr} of $Q$;}
    \STATE Generate the {\initialtestcone} $\PP$ by Algorithm
    \ref{alg:itc};
    \STATE $S_0 := \{ \ \}$;
    \FOR{$i$ from $1$ to $m$} \label{ln:inputredbegin}
    \IF{{\rm RedundancyTest}($\PP, \A_i \y \le \c_i$) = false}
    \STATE $S_0 := S_0 \cup \{ \A_i \y \le \c_i \}$;
    \STATE $\PP := \PP |_{v_1 = 0}$;
    \ENDIF
    \ENDFOR \label{ln:inputredend}
    \FOR{$i$ from $0$ to $n-1$}
    \STATE $S_{i+1}:=\{ \ \}$;
    \FOR{${{\ell}_{\rm pos}} \in S_i$ with positive coefficient of $y_{i+1}$}
    \FOR{${{\ell}_{\rm neg}} \in S_i$ with negative coefficient of $y_{i+1}$}
    \STATE $\ell_{{\rm new}} := \combine({{\ell}_{\rm pos}}, {{\ell}_{\rm neg}}, y_{i+1})$;
    \IF{{\rm RedundancyTest}($\PP, \ell_{{\rm new}}$) = false}
    \STATE $S_{i+1} := S_{i+1} \cup \{ \ell_{{\rm new}} \}$;
    \ENDIF
    \ENDFOR
    \ENDFOR
    \FOR{$\ell \in S_i$ with zero coefficient of $y_{i+1}$}
    \IF{{\rm RedundancyTest}($\PP, \ell$) = false}
    \STATE $S_{i+1} := S_{i+1} \cup \{ \ell \}$;
    \ENDIF
    \ENDFOR
    \STATE $\PP := \PP |_{v_{i+1}=0}$;
    \ENDFOR
    \RETURN $S_0 \cup S_1 \cup \cdots \cup S_n$.
  \end{algorithmic}
\end{algorithm}

%% file: initial_projection_proof.tex
       \begin{proof}
  By Step \ref{ln:a0} of Algorithm \ref{alg:itc}, $[\v^t, \w^t]
  \A_0^{-1} \A = \v^t$ holds whenever $(\v, \w, v_0) \in W$.
  Rewrite  $\v$ as $\v^t = [\v_1^t, \v_2^t]$, where $\v_1$ and $\v_2$ are the
  first $p$ and last $n-p$ variables of $\v$.
  We have $[\v^t,  \w^t] \A_0^{-1}A = \v_1^t$ and
  $[\v^t,  \w^t] \A_0^{-1}B = \v_2^t$.
  Similarly, we have $[\v_{\u}^t, \w_{\u}^t] B_0^{-1} A = \0 $
  and $[\v_{\u}^t, \w_{\u}^t] B_0^{-1} B = \v_{\u}^t$ whenever
  $(\v_{\u}, \w_{\u}, v_{\u}) \in W^0$.
  This lemma holds if we can show $\PP_{\u} = \PP | _{\v_1 = \0 }$.
  We prove this in two steps.

  $(\subseteq)$ For any $(\overline{\v}_{\u}, \overline{v}_{\u}) \in
  \PP_{\u}$, there exists $\overline{\w}_{\u} \in \Q^{m-q}$ satisfying
  $(\overline{\v}_{\u}, \overline{\w}_{\u}, \overline{v}_{\u}) \in W^0$.
  Let $[\overline{\v}^t, \overline{\w}^t] :=
  [\overline{\v}_{\u}^t, \overline{\w}_{\u}^t] B_0^{-1} \A_0$, where
  $\overline{\v}^t = [\overline{\v}_1^t, \overline{\v}_2^t]$ with
  $\overline{\v}_1 \in\Q^{p}, \overline{\v}_2 \in \Q^{n-p}$ and
  $\overline{\w} \in \Q^{m-n}$.
  Then, $\overline{\v}_1^t = [\overline{\v}_{\u}^t, \overline{\w}_{\u}^t]
  B_0^{-1} A = \0 $
  and $\overline{\v}_2^t = [\overline{\v}_{\u}^t, \overline{\w}_{\u}^t]
  B_0^{-1} B = \overline{\v}_{\u} $
  due to $(\overline{\v}_{\u}, \overline{\w}_{\u},
  \overline{v}_{\u}) \in W^0$.
  Let $\overline{v}_0 = \overline{v}_{\u}$, it is easy to verify that
  $(\overline{\v}, \overline{\w}, \overline{v}_0) \in W$.
  Therefore, $(\0 , \overline{\v}_{\u},\overline{v}_{\u}) =
  (\overline{\v},  \overline{v}_0) \in \PP$.

  $(\supseteq)$ For any $(\0 , \overline{\v}_2, \overline{v}_0) \in
  \PP$, there exists $\overline{\w} \in \Q^{m-n}$ satisfying
  $(\0 , \overline{\v}_2, \overline{\w}, \overline{v}_0) \in W$.
  Let $(\overline{\v}_{\u}, \overline{\w}_{\u}) := (\0 ,
  \overline{\v}_2, \overline{\w}) \A_0^{-1}B_0$.
  We have $\overline{\v}_{\u} = (\0 ,
  \overline{\v}_2, \overline{\w}) \A_0^{-1}B = \overline{\v}_2$.
  Let $\overline{v}_{\u} = \overline{v}_0$, it is easy to verify
  $(\overline{\v}_{\u}, \overline{\w}_{\u}, \overline{v}_{\u}) \in
  W^0$.
  Therefore, $(\overline{\v}_2, \overline{v}_0) = (\overline{\v}_{\u},
  \overline{v}_{\u}) \in \PP_{\u}$.
\end{proof}

%% file: complexity.tex
\section{Complexity estimates}
\label{sec:comp}

We analyze the computational complexity of \Cref{alg:mpr}, which
computes the minimal projected representation of a given polyhedron.
This computation is equivalent to eliminate all variables, one after
another, in Fourier-Motzkin elimination.  We prove that using our
algorithm, finding a minimal projected representation of a polyhedron
is singly exponential in the dimension $n$ of the ambient space.

The most consuming procedure in \Cref{alg:mpr} is finding the
{\initialtestcone}, which requires another polyhedron projection in
higher dimension.  As it is shown in Remark~\ref{re:proj}, we can use
block elimination method to perform this task efficiently.  This
requires the computations of the extreme rays of the projection cone.
The double description method is an efficient way to solve this
problem.  We begin this section by computing the bit complexity of the
double description algorithm.

\begin{lemma}[Coefficient bound of extreme rays]
\label{le:extcoefbound}
Let $S = \{ \x \in \Q^n \ | \  A \x \leq \0 \}$ be a
 minimal representation of a cone $C \subseteq \Q^n$, 
where $A \in \Q^{m \times n}$. Then, the 
absolute value of a coefficient in any extreme ray of $C$ 
is bounded over by $(n-1)^n \| A \|^{2(n-1)}$.
\end{lemma}

\iflong
\input{extreme_ray_bound}
\fi

\begin{lemma}
\label{le:ddcomp}
Let $S = \{ \x \in \Q^n \ | \  A\x \leq \0 \}$ be the minimal representation of a cone $C \subseteq
\Q^n$, where $A \in \Q^{m \times n}$.
\ifCoRRversion
The double description method, as specified in \Cref{alg:DDmethod},
\fi
requires $O(m^{n+2} n^{\theta+\epsilon}{h}^{1+\epsilon})$
bit operations, where $h$ is the height of the matrix $A$.
\end{lemma}

\iflong
\input{dd_complexity}

\fi

\begin{lemma}[Complexity of constructing the {\initialtestcone}]
  \label{le:compcrtc}
Let $h$ be the maximum height of $A$ and $\c$ in the input system,
then generating the {\initialtestcone} (Algorithm \ref{alg:itc})
requires at most
$O(m^{n+3+\epsilon}(n+1)^{\theta+\epsilon}h^{1+\epsilon})$ bit
operations.  Moreover, $\proj{W; \{\v, v_0\}}$ can be represented by
$O(m^{\lfloor \frac{ n+1 }{2}\rfloor })$ inequalities, each
with a height bound of $O(m^{\epsilon}n^{2+\epsilon}h)$.
\end{lemma}

\iflong
\input{detailed_comp}

\fi

\begin{lemma}
\label{le:extremeraycheck}
\Cref{alg:checkextremeray} runs within {$O(m^{\frac{n}{2}}n^{\theta+\epsilon}h^{1+\epsilon})$} bit operations.
\end{lemma}

\iflong
\input{check_extreme_ray_proof}
\fi

Using Algorithms~\ref{alg:itc} and \ref{alg:checkextremeray}, 
we can find the {\mpr} of a polyhedron in
singly exponential time w.r.t. the number of variables $n$.

\begin{theorem}\label{thm:comp}
  Algorithm \ref{alg:mpr} is correct.
  Moreover, a minimal projected representation of $Q$ can be produced within
  $O(m^{\frac{5n}{2}} n^{\theta + 1 + \epsilon} h^{1+\epsilon})$ bit operations.
\end{theorem}
\begin{proof}
Correctness of the algorithm follows from \Cref{thm:correspondence}, \Cref{le:initialprojection}.

By \cite{imbert1993fourier,kohler1967projections}, we know that after eliminating $p$ variables, the projection of the 
polyhedron has at most $m^{p+1}$ facets. For eliminating the next variable, there will be at most
$(\frac{m^{p+1}}{2})^2$ pairs of inequalities to be considered and each of the pairs generate a new inequality which should
be checked for redundancy. Therefore, overall the complexity of the algorithm is:
\begin{center}
$O(m^{n+3+\epsilon}(n+1)^{\theta+\epsilon}h^{1+\epsilon}) +
  \sum_{p=0}^{n} m^{2p+2}
  O(m^{\frac{n}{2}}n^{\theta+\epsilon}h^{1+\epsilon}) =
  O(m^{\frac{5n}{2}} n^{\theta + 1 + \epsilon} h^{1+\epsilon}).$
\end{center}
\end{proof}

%% file: extreme_ray_bound.tex
\begin{proof}
From the properties of extreme rays, see
Section~\ref{sec:PolyhedralCones},
\ifCoRRversion
by \Cref{le:algetest},
\fi
we know that when $\r$ is an extreme ray, there
exists a sub-matrix $A' \in \Q^{(n-1) \times n}$ of $A$, such that
$A'\r=0$. This means that $\r$ is in the null-space of $A'$. Thus, the
claim follows by proposition 6.6 of \cite{storjohann2013algorithms}.
\end{proof}

%% file: dd_complexity.tex
\begin{proof}
To analyze the complexity of the DD method 
after adding $t$ inequalities, with $n \leq t \leq m$, 
the first step is to partition the extreme rays 
at the ${t-1}$-iteration, 
with respect to the newly added inequality.
Note that we have at most $(t-1)^{\lfloor{\frac{n}{2}}\rfloor}$ 
extreme rays
\ifCoRRversion
(\Cref{le:maxextr})
\fi
whose coefficients can be bounded
over  by $(n-1)^n \|A\|^{2(n-1)}$ (Lemma~\ref{le:extcoefbound}) 
at the ${t-1}$-iteration.
Hence, this step needs at most
$C_1 := (t-1)^{\lfloor{\frac{n}{2}}\rfloor} \times n \times \MC (\log((n-1)^n \|A\|^{2(n-1)})) \leq O(t^{
\lfloor \frac{n}{2} \rfloor} n^{2+\epsilon} h^{1+\epsilon})$ bit operations. 
After partitioning the vectors, the next
step is to check adjacency for each pair of vectors.
The cost of this step is equivalent to computing the 
rank of a sub-matrix $A' \in \Q^{(t-1) \times n}$ of $A$.
 This should be done for $\frac{t^n}{4}$ pairs of vectors. This step needs at most
$C_2 := \frac{t^n}{4} \times O((t-1)n^{\theta+\epsilon}h^{1+\epsilon}) \leq O(t^{n+1}n^{\theta+\epsilon}h^{1+\epsilon})$
 bit operations.
 \ifCoRRversion
 By \Cref{le:maxextr}, 
 we
 \fi
 know there are at most $t^{\lfloor{\frac{n}{2}}\rfloor}$ pairs of adjacent extreme rays. 
The next step is to combine every pair of adjacent vectors 
in order to obtain a new extreme ray.
This step consists of $n$ multiplications in {\Q} of coefficients with absolute value
bounded over by  $(n-1)^n \|A\|^{2(n-1)}$ 
(Lemma~\ref{le:extcoefbound}) and this should be done 
for at most $t^{\lfloor{\frac{n}{2}}\rfloor}$ vectors. 
Therefore, the bit complexity of this step, is no more than
$C_3 := t^{\lfloor{\frac{n}{2}}\rfloor} \times n \times \MC (\log((n-1)^n \|A\|^{2(n-1)})) \leq
O(t^{\lfloor \frac{n}{2} \rfloor} n^{2+\epsilon} h^{1+\epsilon})$.
Finally, the complexity of step $t$ of the algorithm is $C := C_1+ C_2+C_3$. 
The claim follows after simplifying $m \cdot C$.
\end{proof}

%% file: detailed_comp.tex
\begin{proof}
We analyze Algorithm~\ref{alg:itc} step by step.
 
{\noindent \small \bf  Step 1: construction of $A_0$ from $A$.}
The cost of this step can be neglected. 
However, it should be noticed that the matrix $A_0$ 
has a special structure.
Without loss of generality, we cam 
assume that the first $n$ rows of $A$ are
linearly independent.
The matrix $A_0$ has the following structure 
$A_0$ $=\left( \begin{aligned} A_1 & \quad \0 \\ A_2 & \quad I_{m-n} \end{aligned} \right)$, 
where $A_1$ is a full rank matrix in $\Q^{n \times n}$ and $A_2 \in \Q^{(m-n) \times n}$.

{\noindent \small \bf Step 2: construction of the cone $W$.}
Using the structure of the matrix $A_0$, its inverse can be expressed as
 $A_0^{-1} = \left( \begin{aligned} A_1^{-1} & \quad \0 \\ -A_2A_1^{-1} & \quad
      I_{m-n} \end{aligned} \right)$.
Also, from \Cref{costmodel} we have $\| A_1^{-1} \| \le (\sqrt{n-1} \| A_1 \|)^{n-1}$.
Therefore, $\|A_0^{-1}\| \le n^{\frac{n+1}{2}} \| A \|^q$, and  
$\|A_0^{-1} \c \| \le n^{\frac{n+3}{2}} \| A \|^n \| \c \| + (m-n)\|\c\|$. That is, 
$\height(A_0^{-1}) \in O(n^{1+\epsilon}h)$ and $\height(A_0^{-1} \c) \in O(m^{\epsilon} + n^{1+\epsilon}h)$.
As a result, height of coefficients of $W$ can be bounded over 
by $O(m^{\epsilon} + n^{1+\epsilon}h)$.

To estimate the bit complexity, we need the following consecutive steps:
\begin{itemize}
\item[-] Computing $A_0^{-1}$, 
which requires 
\begin{center} 
$\begin{aligned} &O(n^{\theta + 1 + \epsilon}h^{1+\epsilon})
 +O((m-n) n^2 \MC(\max(\height(A_2), \height(A_1^{-1}))))\\ \le & O(mn^{\theta + 1 + \epsilon}h^{1+\epsilon}) \text{ bit
 operations;} \end{aligned}$ 
 \end{center}
\item[-]Constructing $W:= \{ (\v, \w, v_0 ) \ | \ - [\v^t, \w^t]\A_0^{-1} \c + v_0 \ge
0,  [\v^t, \w^t]\A_0^{-1} \ge \0 \}$ requires at most
\begin{center} 
 $\begin{aligned}
C_1:= & O(m^{1+\epsilon}n^{\theta + 1 +
  \epsilon}h^{1+\epsilon})+{O(m n
  \MC(\height(A_0^{-1},\c)))}\\
+ & O((m-n)h) \le  O(m^{1+\epsilon}n^{\theta+\epsilon+1}h^{1+\epsilon}) \text{ bit operations.}
\end{aligned}$ 
\end{center}           
\end{itemize}

{\noindent \small \bf Step 3: projecting $W$ and finding the {\initialtestcone}.}
Following Lemma~\ref{le:projectwithextremeray}, we obtain a
representation of $\proj{W; \{\v, v_0\}}$ through finding extreme rays of the 
corresponding projection cone. 

\noindent
Let $E = (-A_2A_1^{-1})^t \in \Q^{n \times (m-n)}$ and $\g^t$ be the last $m-n$ elements of
$(A_0^{-1}\c)^t$. Then, the projection cone can be represented by:
\begin{center}
$C =\{ \y \in\Q^{m+1} \ | \ \y^t \left( \begin{aligned} E
\\ \g^t \\ I_{m-n} \end{aligned}
         \right) = \0 , \y \ge \0 \}.$
\end{center}
Note that $y_{n+2}, \ldots, y_{m+1}$ can be solved from the system of equations
in the representation of $C$. We substitute them in the inequalities and obtain 
a representation of the cone $C'$, given by:
\begin{center}
$C' = \{ \y' \in \Q^{n+1}
 \ | \ \y'^t \left( \begin{aligned} E \\ \g^t  \end{aligned}\right) \leq \0 , \y' \ge \0 \}$
\end{center}
In order to find the extreme rays of the cone $C$, we can find 
the extreme rays of the cone $C'$
and then back-substitute them into the equations to find the extreme rays of $C$.
Applying \Cref{alg:DDmethod} to $C'$, we can obtain all extreme rays
of $C'$, and subsequently, the extreme rays of $C$. 
The cost estimate of this step is bounded over by
the complexity of \Cref{alg:DDmethod} with $C'$ as input.
This operation requires at most
    $C_2 := O(m^{n+3}
    (n+1)^{\theta+\epsilon} \max(\height(E,\g^t))^{1+\epsilon})
    \le O(m^{n+3+\epsilon}
    (n+1)^{\theta+\epsilon}h^{1+\epsilon})$ bit operations.
The overall complexity of the algorithm can be bounded over by:
$C_1+C_2 \le  O(m^{n+3+\epsilon}(n+1)^{\theta+\epsilon}h^{1+\epsilon}).$
Also, by \Cref{le:extcoefbound} and \Cref{le:ddcomp}, we know that
the cone $C$ has at most $O(m^{\lfloor \frac{n+1}{2} \rfloor})$ distinct extreme
rays, each with height no more than $O(m^{\epsilon}n^{2+\epsilon}h)$.
 That is,  $\proj{W^0; \{\v, v_0\}}$ can be represented by at most $O(m^{\lfloor \frac{ n+1 }{2}\rfloor })$ inequalities,
each with a height bound of $O(m^{\epsilon}n^{2+\epsilon}h)$.
\end{proof}

%% file: check_extreme_ray_proof.tex
\begin{proof}
The first step is to multiply the matrix $M$ and the vector $(\t,t_0)$. 
Let $d_M$ and $c_M$ be the number of rows and columns of $M$,
respectively, thus $M \in \Q^{d_M \times c_M}$. 
We know that $M$ is the coefficient matrix of $\proj{W^0, \{ \v, v_0 \}}$. 
Therefore, after eliminating $p$ variables $c_M = q+1$, where $q = n-p$ and
$d_M \leq m^{\frac{n}{2}}$. 
Also, we have $\height(M) \in O(m^{\epsilon}n^{2+\epsilon}h)$. 
With these specifications, the multiplication step and the rank computation step need 
$O(m^{\frac{n}{2}}n^{2+\epsilon}h^{1+\epsilon})$ and $O(m^{\frac{n}{2}}(q+1)^{\theta+\epsilon}h^{1+\epsilon})$
bit operations, respectively, and the claim follows after simplification.
\end{proof}

%% file: Experiments.tex
\ifCoRRversion
\section{Experimentation}
\label{sec:exp}

This section reports on our software implementation
of the algorithms presented in the previous sections.
Our code is part of the BPAS library,
which is available at \url{www.bpaslib.org}
and is written 
in the \textsc{C} programming language.  
We tested our algorithm in terms of effectiveness for removing redundant inequalities and also in terms of running time.
The first thirteen test cases, (t1 to t13) are linear inequality 
systems with random coefficients; moreover, of these systems
is consistent, that is, has a non-empty solutiiion set.
The systems S24 and S35 are 24-simplex and 35-simplex polytopes, 
C56 and C510 are cyclic polytopes 
in dimension five with six and ten vertices,  
C68 is a cyclic polytope in dimension six with eight vertices and C1011
is cyclic polytope in dimension ten with eleven vertices~\cite{henk200416}. 
Our test cases can be found at \url{www.bpaslib.org/FME-tests.tgz}.
In our implementation, each system of linear inequalities is encoded
by an unrolled linked list, where each cell stores an inequality in a
dense representation.

\Cref{number} illustrates the effectiveness of each redundancy elimination method. The columns \texttt{\#var} and \texttt{\#ineq}
specify the number of variables and inequalities of each input system,
respectively. The last two columns show the maximum number of
inequalities appearing in the process of FME algorithm. The
column \texttt{check1} corresponds to the case that the Kohler's
algorithm is the only method for redundancy detection and the
column \texttt{check2} is for the case that Balas' algorithms is used.
Column \texttt{MinProjRep} gives the running times of our algorithm
for computing a minimal projected representation.

The \texttt{Maple} column  shows the running time 
for the \texttt{Projection} function of the \texttt{PolyhedralSets} package in Maple.
The last two columns show running time of Fourier elimination function in the 
\textsc{CDD} library.
The \texttt{CDD1} column is running time of the function when it uses an 
LP method for redundancy elimination, while the \texttt{CDD2} column is the running time of the same function  but it uses
Clarkson's algorithm \cite{clarkson1994more}.
\footnote{Because the running time of the
algorithm for eliminating all variables is more than one hour for some cases we only remove \textit{some} of the variables. The numbers in level parts shows the number of variables that can be eliminated in one hour of running program}

\begin{table}
\footnotesize{
 \begin{tabular}{|c | c | c  | c |c |} 
 \hline
 Test case & \# var & \# ineq & check 1 & check 2\\
 \hline
 \hline
 t1   & 5 & 10 & 36 & 20 \\  
 \hline
 t2  & 10 & 12 & 73 & 66 \\  
 \hline
 t3  & 4 & 8 & 20 & 11   \\ 
 \hline
 t4 & 5 & 10 & 33 & 19   \\ 
 \hline
 t5 & 5 & 8 & 20 & 14   \\ %
 \hline
 t6 & 7 & 10 & 40 & 37 \\    
 \hline
 t7  & 10 & 12 & 92 & 82 \\   
 \hline
 t8 & 6 & 8 & 18 & 15 \\  
 \hline
 t9 & 5 & 11 & 52 & 18 \\  
 \hline
 t10 & 10 & 20 &1036 & 279\\   
 \hline
 t11 & 9 & 19 & 695 & 362 \\
 \hline
 t12 & 8 & 19 & 620 & 257 \\
 \hline
 t13 & 6 & 18 & 435 & 91 \\
 \hline
 S24 & 24 & 25 & 24  & 24   \\
 \hline
 S35 & 35 & 36 & 35  &  35  \\
 \hline
 C56 & 5 & 6 & 9  &  9   \\
 \hline
 C68 & 6 & 16 & 24   &  20   \\
 \hline
 C1011 & 10 & 11 & 77  &  77  \\
 \hline
 C510 & 5 & 42 & 24024   & 35    \\
 \hline
 \end{tabular}
 }
 \caption{Maximum number of inequalities }
 \label{number}
\end{table}

\begin{table}
\footnotesize{
 \begin{tabular}{|c | c| c | c | c  | c | c |} 
 \hline
 Case   & MinProjRep & Maple & CDD1  & CDD2 \\ [0.5ex] 
 \hline\hline
 t1  &  8.042 & 7974 & 142 & 47 \\  
 \hline
 t2  & 107.377 & 3321217 & 122245 & 7925  \\  
 \hline
 t3 & 2.193  & 736 & 4  &  1\\ 
 \hline
 t4 & 5.960 & 2579& 48 &  17\\ 
 \hline
 t5 &  3.946 & 3081 & 32 & 13\\ 
 \hline
 t6  &  26.147 & 117021 & \small{core dump} & \footnotesize{wrong result}\\  
 \hline
 t7 &   353.588 & $>$1h &  1177807& 57235 \\   
 \hline
 t8 & 4.893 &  4950 & 124 & 22 \\   
 \hline
 t9 & 8.858 & 8229 & 75   &  39  \\   
 \hline
 t10 & 24998.501 & $>$ 1h & $>$ 1h (2)& $>$1h (3)\\  
 \hline    
 t11& 191191.909 & $>$ 1h & $>$ 1h (2)& $>$ 1h (2)\\ 
 \hline
 t12 & 21665.704 & $>$ 1h &  $>$1h (2)&  746581\\ 
 \hline
 t13 &  1264.289 &  $>$ 1h & 77372 & 30683 \\ 
 \hline
 S24 & 39.403 & 6485 & 334    &   105     \\ 
 \hline
 S35 & 158.286 & 57992 & 1827    &   431    \\ 
 \hline
 C56 & 1.389 & 825  &  11      &   3         \\ 
 \hline
 C68 & 4.782 &  20154   & 682  &  75     \\ 
 \hline
 C1011& 85.309 & $>$ 1h & $>$1h (4) &  76861  \\ 
 \hline
 C510&  23.973  & 6173 & 6262 & 483  \\ 
 \hline
 \end{tabular}}
 \caption{Running time comparison (ms)}
 \label{runtime}
\end{table}

\fi

%% file: Relatedwork.tex
\section{Related work}
\label{sec:relatedwork}

During our study of the {\FM}, we found many related works.  As
discussed above, removing redundant inequalities during the execution
of {\FM} is the central issue towards efficiency.
To our knowledge, all available implementations of {\FM}
rely on linear programming for removing all the redundant 
inequalities, an idea suggested in~\cite{DBLP:reference/opt/Khachiyan09}.
However, and as mentioned above, there are alternative 
algorithmic approaches
relying on linear algebra.
In~\cite{chernikov1960contraction}, 
Chernikov proposed a redundancy test with
little added work, which greatly improves the practical efficiency
of {\FM}.
Kohler proposed a method in \cite{kohler1967projections} which only uses
matrix arithmetic 
operations to test the redundancy of inequalities.
As observed by Imbert in his work~\cite{imbert1993fourier}, 
the method he proposed in this paper as well as those 
of Chernikov and Kohler  are essentially equivalent.
Even though these works are very effective in practice, 
none of them can remove all redundant inequalities
generated by {\FM}.

Besides {\FM}, block elimination is another algorithmic tool to
project polyhedra on a lower dimensional subspace.
This method relies on the extreme rays
of the so-called {\projectcone}. 
Although there exist efficient methods to enumerate the extreme rays
of this {\projectcone}, like the 
{\em double description method}~\cite{fukuda1996double}
(also known as Chernikova's algorithm~\cite{chernikova1965algorithm,le1992note}), this method can not remove all the redundant inequalities.

In~\cite{balas1998projection}, Balas shows that if certain
{\em inconvertibility conditions} are satisfied, then the extreme rays of the
{\testcone} exactly defines a minimal representation of the
projection of a polyhedron. As Balas mentioned in his paper, this method can
be extended to any polyhedron.
Through experimentation, we found that the results and constructions 
in Balas' paper had some flaws.
First of all, in Balas' work, the redundancy test cone is defined as the projection of the cone
$
W^0 := \{ (\v , \w , v_0)\in
  \Q^q \times \Q^{m-q} \times \Q \ 
| \ [\v^t, \w^t]B_0^{-1}A_0 = 0 , 
-[\v^t, \w^t]B_0^{-1}\c_0+v_0 = 0
, [\v^t, \w^t] B_0^{-1} \geq 0 \}
$
on the $(\v,v_0)$ space.
The Author claimed that 
$\a^t \x \le c$
defines a facet of the projection $\proj{Q; \x}$ if and only if
$(\a, c)$ is an extreme ray of
the {\testcone} $\proj{W^0; \{ \v, v_0\}}$.
However, we have a counter example for this claim. 
\iflong
Please refer to the page \url{http://www.jingrj.com/worksheet.html}.
In this example, when we eliminate two variables, 
the cone $\proj{W^0; \{ \v, v_0 \}}$
has 19 extreme rays while $\proj{Q;\x}$ has 18 facets.
18 of the 19 extreme rays of $\proj{W^0; \{ \v, v_0 \}}$ give out the 18
facets of $\proj{Q;\x}$, while the remaining extreme ray gives out a
redundant inequality w.r.t. the 18 facets.
\fi
The main reason leading to this situation is due 
to a misuse of Farkas' lemma in the
proof of Balas' paper.
We improved this situation 
by changing $-[\v^t, \w^t]B_0^{-1}\c_0+v_0 = 0$ to
$-[\v^t, \w^t] B_0^{-1}\c_0+v_0 \ge 0$ and carefully showed the relations 
between the 
extreme rays of $\proj{W^0; \{ \v, v_0 \}}$ and the facets of $\proj{Q;\x}$, for
the details please refer to
Theorems~\ref{thm:polarcone}, \ref{thm:correspondence}.
In fact, with our change in the construction of $W^0$, we will have at most one
extra extreme ray, which is always $(\0 ,1)$.
An other drawback of Balas' work is 
that the necessity of enumerating the extreme rays of
  the redundancy test cone in order to produce a minimal
  representation of $\proj{Q;\x}$, which is time consuming.
 Our algorithm tests 
  the redundancy of the inequality $\a
  \x \le c$ by checking whether $(\a, c)$ is an
  extreme ray of the redundancy test cone or not.

\ifCoRRversion
\input{Subsumption_cone}
\fi

\ifCoRRversion
Based on the improved version of Balas' methods, 
we obtain an
algorithm to remove all the redundant inequalities produced by {\FM}.
Even though this algorithm still has exponential complexity, 
which is expected, 
it is very 
effective in practice, as we have shown in
Section~\ref{sec:exp}.

The projection of polyhedra is a useful tool to solve 
problem instances in 
parametric
linear programming, which plays an important role in the analysis,
transformation and scheduling of for-loops of computer programs,
see for instance~\cite{jones2005polyhedral,borrelli2003geometric,jones2008polyhedral}.
\fi

%% file: Subsumption_cone.tex
\subsection{Subsumption Cone}\label{sec:sc}
After revisiting Balas' method, we found another cone called {\sc} \cite{huynh1992practical,lassez1990querying},
which we will prove later equals to the initial test cone $\PP :=\proj{W;\{\v, v_0\}}$ in the previous section.

Consider the polyhedron $Q$ given in Equation (\ref{eq:Q}), denote
$T := \{ ( \bflambda,  \bfalpha, \beta) \ | \ \bflambda^t \A =
\bfalpha^t, \bflambda^t \c \le \beta, \bflambda \ge \0 \}$,
where $\bflambda$ and $\bfalpha$ are column vectors of dimension $m$
and $n$ respectively, $\beta$ is a variable.
The {\em {\sc}} of $Q$ is obtained by eliminating $\bflambda$ in $T$,
that is, $\proj{T; \{ \bfalpha, \beta \}}$.

Remember that we can obtain the  initial test cone $\PP =
\proj{W;\{\v, v_0 \}}$ by  Algorithm \ref{alg:itc}, here
$W := \{ (\v, \w, v_0 ) \ | \ - [\v^t, \w^t]\A_0^{-1} \c + v_0 \ge
0,  [\v^t, \w^t]\A_0^{-1} \ge \0 \}$.

\import{}{subsumption_cone_proof}

In Section \ref{sec:MRPP}, we have shown how to use the initial test
cone to remove all the redundant inequalities and give a minimal
representation of the projections of given pointed polyhedra.
Detailed proofs are also explained in the previous section.
It also applies to the subsumption cone.
In \cite{huynh1992practical,lassez1990querying}, the authors mentioned
that the {\sc} can not detect all the redundant inequalities.
However, their object is full-dimensional
polyhedra while ours are pointed polyhedra.
Notice that any full-dimensional polyhedron can be transformed to a
pointed polyhedron by some coordinate transformations.

%% file: subsumption_cone_proof.tex
\begin{lemma}
  \label{le:scequaltoitc}
  The {\sc} of $Q$ equals to its {\initialtestcone} $\PP$.
\end{lemma}
\begin{proof}
  Let $\bflambda^t := [\overline{\v}^t, \overline{\w}^t] \A_0^{-1}$ and $\beta = \overline{v}_0$, we
  prove the lemma in two steps.

  $(\subseteq)$ For any $(\bfalpha, \beta)$ in the {\sc} $\proj{T;
    \{\bfalpha, \beta\}}$, there exists $\bflambda \in \Q^m$ satisfying
  $(\bflambda, \bfalpha, \beta) \in T$.
  Remember that $\A_0 =
       [\A, \A']$, where $\A'= [\e_{n+1}, \ldots, \e_{m}]$ with $\e_i$
       being the $i$-th canonical basis of $\Q^n$ for $i: n+1 \le i
       \le m$, we have $\A_0^{-1} \A = [\e_1, \ldots, \e_{n}]$ with
       $\e_i$ being the $i$-th canonical basis of $\Q^n$ for $i: 1\le
       i \le n$.
       Hence, $\bfalpha^t = \bflambda^t \A = [\overline{\v}^t,\overline{\w}^t] \A_0^{-1} \A
       = \overline{\v}^t$.
       Also, we have $[\overline{\v}^t, \overline{\w}^t] \A_0^{-1} \c \le
       \beta=\overline{v}_0, [\overline{\v}^t, \overline{\w}^t] \A_0^{-1} \ge \0 $.
       Therefore, $(\bfalpha, \beta) = (\overline{\v}, \overline{v}_0) \in \proj{W;\{\v,
         v_0\}}$.

       $(\supseteq)$ For any $(\overline{\v}, \overline{v}_0)$ in the {\initialtestcone}
       $\proj{W; \{\v, v_0\}}$, there exists $\overline{\w}\in
       \Q^{m-n}$ satisfying $(\overline{\v}, \overline{\w}, \overline{v}_0) \in \proj{W;\{\v, v_0\}}$.
       Let $\bfalpha = \overline{\v}$.
       Then, $\bfalpha^t = \overline{\v}^t = [\overline{\v}^t,\overline{\w}^t] \A_0^{-1} \A =
       \bflambda^t \A$,
       $\bflambda^t \c = [\overline{\v}^t,\overline{\w}^t] \A_0^{-1} \c \le \overline{v}_0 = \beta$ and
       $\bfalpha^t = [\overline{\v}^t,\overline{\w}^t] \A_0^{-1} \ge \0 $.
       Therefore, $(\overline{\v}, \overline{v}_0) = (\bfalpha, \beta) \in \proj{T; \{
         \bfalpha, \beta \}}$.
\end{proof}

%% file: Application.tex
\section{Solving parametric linear programming problem with \FM}
\label{sec:application}

In this section, we show how to use \FM \  for solving parametric
linear programming (PLP) problem instances.

Given a PLP problem instance:
\begin{equation} \label{eq:PLP}
  \begin{aligned}
     z(\bftheta)  &=  \min \c \x \\
     A \x& \le B \bftheta + \b
      \end{aligned}
\end{equation}
  where $A \in \Z^{m \times n}, B \in \Z^{m \times p}, \b \in
    \Z^{m},$ and $\x \in \Q^{n}$ are the variables, $\bftheta \in
    \Q^{p}$ are the parameters.

To solve this problem, first we need the following preprocessing
step.
Let $g>0$ be the greatest common divisor of elements in $\c$.
Via Gaussian elimination, we can obtain a uni-modular matrix $U \in \Q^{n
  \times n} $ satisfying $[0, \ldots, 0, g] =  \c U$.
Let $\t = U^{-1} \x$, the above PLP problem can be transformed to the
following equivalent form:
\begin{equation} 
\label{eq:PLPtransform}
  \begin{aligned}
    & z(\bftheta) =  \min g t_n \\
    & AU \t \le B \bftheta + \b .
  \end{aligned}
\end{equation}
Applying \Cref{alg:mpr} to the constraints $AU \t \le B \bftheta + \b$
with the variable order $t_1 > \cdots > t_n > \bftheta$, we obtain
${\PR(Q; t_1 > \cdots > t_n > \bftheta)}$, where $Q\subseteq \Q^{n+p}$ is the polyhedron
represented by $AU \t \le B \bftheta + \b$.
We extract the representation of the projection $\proj{Q; \{t_n, \bftheta\}}$,
denoted by $\Phi := \Phi_1 \cup \Phi_2$. Here we denote by $\Phi_1$
the set of inequalities which have a non-zero coefficient in $t_n$ and
$\Phi_2$ the set of inequalities which are free of  $t_n$.
Since $g>0$, 
to solve (\ref{eq:PLPtransform}),
 we only need to consider
the lower bound of $t_n$, which is very easy to deduce from $\Phi_1$.

Consider Example 3.3 in \cite{borrelli2003geometric}:
$$ \begin{aligned}
  &\min \quad  -2 x_1 -x_2\\
  & \left \{
  \begin{aligned}
  & x_1 + 3x_2 \le 9 - 2 \theta_1 + \theta_2,
   2x_1 + x_2 \le 8 + \theta_1 -2 \theta_2 \\
 & x_1 \le 4 + \theta_1 + \theta_2,\quad
 - x_1 \le 0,\quad
 - x_2 \le 0
  \end{aligned}
  \right.
  \end{aligned}$$
We have $(-2, -1) U = (0, 1)$, where $U = \left( \begin{aligned} 1 &
  \quad 0 \\ -2 & \quad -1 \end{aligned} \right)$.
Let $(t_1, t_2)^T = U^{-1} (x_1, x_2)^T$, the above PLP problem is
equivalent to
$$\begin{aligned}
  & \min \quad  t_2 \\
  & \left \{
  \begin{aligned}
  & -5 t_1 - 3 t_2 \le 9 - 2 \theta_1 + \theta_2, -t_2 \le 8 +
    \theta_1 -2 \theta_2 \\ 
 & t_1 \le 4 + \theta_1 + \theta_2, \quad -t_1 \le 0, \quad 2t_1 + t_2 \le 0
  \end{aligned}
  \right.
  \end{aligned}$$
Let $P$ denote the polyhedron represented by the above constraints.
Applying \Cref{alg:mpr} to $P$ with variable order
$t_1>t_2>\theta_1>\theta_2$, we obtain the projected representation
$\PR(P; t_1>t_2>\theta_1>\theta_2)$, from which we can easily extract
the representation of the projected polyhedron $\proj{P; \{t_2, \theta_1, \theta_2\}}$:
$$\proj{P; \{t_2, \theta_1, \theta_2\}} := \left \{ \begin{aligned}
  -t_2- \theta_1+2 \theta_2 & \le 8, \
  -3 t_2 - 3 \theta_1 - 6 \theta_2 &\le 29, \\
  - t_2 + 4 \theta_1 - 2 \theta_2 &\le 18, 
  \quad t_2 & \le 0, \\
  - \theta_1 - \theta_2 & \le 4, \quad \quad \quad -\theta_1 + 2 \theta_2 &\le 8,\\
  -3 \theta_2 &\le 17, \quad \quad \quad  \quad \quad \quad 3 \theta_2 &\le 25.
\end{aligned} \right.
$$
$t_2$ has three lower bounds: $t_2 = -8 - \theta_1 + 2
  \theta_2,
t_2 = -\theta_1 - 2 \theta_2 -29/3$ and $t_2 = 4 \theta_1 - 2 \theta_2
-18$, under the constraints
  {\small \begin{center}
    $\left\{ \begin{aligned}
      & -\theta_2 \le 5/12,  & -\theta_1 - \theta_2 \le 4,\\
      & \theta_1 + 2 \theta_2 \le 8, & \theta_1 - 4/5 \theta_2 \le 2.
    \end{aligned}
    \right.
    $, 
    $ \left\{ \begin{aligned}
      & \theta_2 \le -5/12,   \theta_1 \le 5/3,\\
      & -\theta_1 - \theta_2 \le 4.
    \end{aligned}
    \right.
    $, 
    $ \left\{ \begin{aligned}
      & -\theta_1 \le -5/3,  -\theta_1 + 4/5 \theta_2 \le -2, \\
      & \theta_1 - \theta_2/2 \le 9/2
    \end{aligned}
    \right.
    $
  \end{center}
}

%% file: main.bbl
\begin{thebibliography}{10}

\bibitem{balas1998projection}
Egon Balas.
\newblock Projection with a minimal system of inequalities.
\newblock {\em Computational Optimization and Applications}, 10(2):189--193,
  1998.

\bibitem{Bastoul:2004:CGP:1025127.1025992}
C.~Bastoul.
\newblock Code generation in the polyhedral model is easier than you think.
\newblock In {\em Proceedings of the 13th International Conference on Parallel
  Architectures and Compilation Techniques}, PACT '04, pages 7--16, Washington,
  DC, USA, 2004. IEEE Computer Society.

\bibitem{Benabderrahmane:2010:PMM:2175462.2175484}
M.~Benabderrahmane, L.~Pouchet, A.~Cohen, and C.~Bastoul.
\newblock The polyhedral model is more widely applicable than you think.
\newblock In {\em Proceedings of the 19th joint European conference on Theory
  and Practice of Software, international conference on Compiler Construction},
  CC'10/ETAPS'10, pages 283--303, Berlin, Heidelberg, 2010. Springer-Verlag.

\bibitem{Bondhugula:2008:PAP:1379022.1375595}
U.~Bondhugula, A.~Hartono, J.~Ramanujam, and P.~Sadayappan.
\newblock A practical automatic polyhedral parallelizer and locality optimizer.
\newblock {\em SIGPLAN Not.}, 43(6):101--113, June 2008.

\bibitem{borrelli2003geometric}
Francesco Borrelli, Alberto Bemporad, and Manfred Morari.
\newblock Geometric algorithm for multiparametric linear programming.
\newblock {\em Journal of optimization theory and applications},
  118(3):515--540, 2003.

\bibitem{DBLP:conf/icms/ChenCMMXX14}
Changbo Chen, Svyatoslav Covanov, Farnam Mansouri, Marc {Moreno Maza}, Ning
  Xie, and Yuzhen Xie.
\newblock The basic polynomial algebra subprograms.
\newblock In {\em Mathematical Software - {ICMS} 2014 - 4th International
  Congress, Seoul, South Korea, August 5-9, 2014. Proceedings}, pages 669--676,
  2014.

\bibitem{chernikov1960contraction}
Sergei~Nikolaevich Chernikov.
\newblock Contraction of systems of linear inequalities.
\newblock {\em Doklady Akademii Nauk SSSR}, 131(3):518--521, 1960.

\bibitem{chernikova1965algorithm}
Natal'ja~V. Chernikova.
\newblock Algorithm for finding a general formula for the non-negative
  solutions of a system of linear inequalities.
\newblock {\em Zhurnal Vychislitel'noi Matematiki i Matematicheskoi Fiziki},
  5(2):334--337, 1965.

\bibitem{clarkson1994more}
Kenneth~L Clarkson.
\newblock More output-sensitive geometric algorithms.
\newblock In {\em Foundations of Computer Science, 1994 Proceedings., 35th
  Annual Symposium on}, pages 695--702. IEEE, 1994.

\bibitem{Feautrier91dataflowanalysis}
P.~Feautrier.
\newblock Dataflow analysis of array and scalar references.
\newblock {\em International Journal of Parallel Programming}, 20, 1991.

\bibitem{Feautrier:1996:APP:647429.723579}
Paul Feautrier.
\newblock Automatic parallelization in the polytope model.
\newblock In {\em The Data Parallel Programming Model: Foundations, HPF
  Realization, and Scientific Applications}, pages 79--103, London, UK, UK,
  1996. Springer-Verlag.
\newblock \url{http://dl.acm.org/citation.cfm?id=647429.723579}.

\bibitem{fukudacdd}
Komei Fukuda.
\newblock The {CDD} and {CDD}plus homepage.
\newblock \url{https://www.inf.ethz.ch/personal/fukudak/cdd_home/}.

\bibitem{fukuda1996double}
Komei Fukuda and Alain Prodon.
\newblock Double description method revisited.
\newblock In {\em Combinatorics and computer science}, pages 91--111. Springer,
  1996.

\bibitem{grosser.11.impact}
T.~Grosser, H.~Zheng, R.~Aloor, A.~Simb{\"u}rger, A.~Gr{\"o}{\ss}linger, and
  L.~Pouchet.
\newblock Polly - polyhedral optimization in llvm.
\newblock In {\em First International Workshop on Polyhedral Compilation
  Techniques (IMPACT'11)}, Chamonix, France, April 2011.

\bibitem{henk200416}
Martin Henk, J{\"u}rgen Richter-Gebert, and G{\"u}nter~M Ziegler.
\newblock 16 basic properties of convex polytopes.
\newblock {\em Handbook of discrete and computational geometry}, pages
  255--382, 2004.

\bibitem{huynh1992practical}
Tien Huynh, Catherine Lassez, and Jean-Louis Lassez.
\newblock Practical issues on the projection of polyhedral sets.
\newblock {\em Annals of mathematics and artificial intelligence},
  6(4):295--315, 1992.

\bibitem{imbert1993fourier}
Jean-Louis Imbert.
\newblock Fourier's elimination: Which to choose?
\newblock In {\em PPCP}, pages 117--129, 1993.

\bibitem{jing2017integerpoints}
Rui{-}Juan Jing and Marc {Moreno Maza}.
\newblock Computing the integer points of a polyhedron, {I:} algorithm.
\newblock In {\em Computer Algebra in Scientific Computing - 19th International
  Workshop, {CASC} 2017, Beijing, China, September 18-22, 2017, Proceedings},
  pages 225--241, 2017.

\bibitem{jing2017computing}
Rui{-}Juan Jing and Marc {Moreno Maza}.
\newblock Computing the integer points of a polyhedron, {II:} complexity
  estimates.
\newblock In {\em Computer Algebra in Scientific Computing - 19th International
  Workshop, {CASC} 2017, Beijing, China, September 18-22, 2017, Proceedings},
  pages 242--256, 2017.

\bibitem{jones2005polyhedral}
Colin Jones.
\newblock Polyhedral tools for control.
\newblock Technical report, University of Cambridge, 2005.

\bibitem{jones2008polyhedral}
Colin~N. Jones, Eric~C. Kerrigan, and Jan~M. Maciejowski.
\newblock On polyhedral projection and parametric programming.
\newblock {\em Journal of Optimization Theory and Applications},
  138(2):207--220, 2008.

\bibitem{DBLP:reference/opt/Khachiyan09}
Leonid Khachiyan.
\newblock Fourier-motzkin elimination method.
\newblock In Christodoulos~A. Floudas and Panos~M. Pardalos, editors, {\em
  Encyclopedia of Optimization, Second Edition}, pages 1074--1077. Springer,
  2009.

\bibitem{kohler1967projections}
David~A. Kohler.
\newblock Projections of convex polyhedral sets.
\newblock Technical report, California Univ. at Berkeley, Operations Research
  Center, 1967.

\bibitem{lassez1990querying}
Jean-Louis Lassez.
\newblock Querying constraints.
\newblock In {\em Proceedings of the ninth ACM SIGACT-SIGMOD-SIGART symposium
  on Principles of database systems}, pages 288--298. ACM, 1990.

\bibitem{le1992note}
Herv{\'e} Le~Verge.
\newblock {\em A note on Chernikova's algorithm}.
\newblock PhD thesis, INRIA, 1992.

\bibitem{mcmullen1970maximum}
Peter McMullen.
\newblock The maximum numbers of faces of a convex polytope.
\newblock {\em Mathematika}, 17(2):179--184, 1970.

\bibitem{monniaux2010quantifier}
David Monniaux.
\newblock Quantifier elimination by lazy model enumeration.
\newblock In {\em International Conference on Computer Aided Verification},
  pages 585--599. Springer, 2010.

\bibitem{DBLP:journals/computing/SchonhageS71}
Arnold Sch{\"{o}}nhage and Volker Strassen.
\newblock Schnelle multiplikation gro{\ss}er zahlen.
\newblock {\em Computing}, 7(3-4):281--292, 1971.

\bibitem{Schrijver:1986:TLI:17634}
Alexander Schrijver.
\newblock {\em Theory of linear and integer programming}.
\newblock John Wiley \& Sons, Inc., New York, NY, USA, 1986.

\bibitem{storjohann2013algorithms}
Arne Storjohann.
\newblock {\em Algorithms for matrix canonical forms}.
\newblock PhD thesis, Swiss Federal Institute of Technology Zurich, 2000.

\bibitem{terzer2009large}
Marco Terzer.
\newblock {\em Large scale methods to enumerate extreme rays and elementary
  modes}.
\newblock PhD thesis, ETH Zurich, 2009.

\bibitem{DBLP:journals/taco/VerdoolaegeJCGTC13}
S.~Verdoolaege, J.~{Carlos Juega}, A.~Cohen, J.~{Ignacio G{\'{o}}mez},
  C.~Tenllado, and F.~Catthoor.
\newblock Polyhedral parallel code generation for {CUDA}.
\newblock {\em {TACO}}, 9(4):54, 2013.

\bibitem{waldschmidt2000diophantine}
Michel Waldschmidt.
\newblock {\em Diophantine approximation on linear algebraic groups}.
\newblock Springer Verlag, 2000.

\end{thebibliography}
